\documentclass[12pt]{article}

\usepackage{enumerate, color, multirow}
\usepackage{amsbsy, amstext, amssymb, amsthm, amsmath,bm}
\usepackage{graphicx,booktabs,natbib}
\usepackage{algorithmic,algorithm}
\usepackage[total={6.2in,8.75in}]{geometry}
\usepackage{array,stmaryrd}
\usepackage{hyperref}

\newtheorem{lemma}{{\bf Lemma}}

\newtheorem{theorem}{{\bf Theorem}}

\newcommand{\vect}{\mathrm{vec}}

\newcommand{\real}[1]{\mathrm{I \! R} \mathit{^{#1}}}
\newcommand{\trans}{^{\mbox{\tiny {\sf T}}}}

\newcommand{\Abf}{{\bm A}}
\newcommand{\Bbf}{{\bm B}}
\newcommand{\Dbf}{{\bm D}}
\newcommand{\Fbf}{{\bm F}}

\newcommand{\Hbf}{{\bm H}}
\newcommand{\Ibf}{{\bm I}}
\newcommand{\Jbf}{{\bm J}}
\newcommand{\Kbf}{{\bm K}}

\newcommand{\Mbf}{{\bm M}}

\newcommand{\Rbf}{{\bm R}}

\newcommand{\Vbf}{{\bm V}}

\newcommand{\Xbf}{{\bm X}}
\newcommand{\Ybf}{{\bm Y}}
\newcommand{\Zbf}{{\bm Z}}

\newcommand{\bbf}{{\bm b}}

\newcommand{\sbf}{{\bm s}}
\newcommand{\ubf}{{\bm u}}

\newcommand{\zerobf}{{\mathbf 0}}

\newcommand{\greekbold}[1]{\mbox{\boldmath $#1$}}

\newcommand{\betabf}{\greekbold{\beta}}

\newcommand{\gammabf}{\greekbold{\gamma}}
\newcommand{\mubf}{\greekbold{\mu}}

\newcommand{\Pibf}{\greekbold{\Pi}}

\title{Tensor Generalized Estimating Equations\\
for Longitudinal Imaging Analysis}
\author{Xiang Zhang, Lexin Li, Hua Zhou, Dinggang Shen\\
and the Alzheimer's Disease Neuroimaging Initiative}
\date{}

\begin{document}
\maketitle

\begin{footnotetext}[1]
{Xiang Zhang is Graduate Student, Department of Statistics, North Carolina State University, Raleigh, NC 27695-8203 (Email: xzhang23@ncsu.edu). 
Lexin Li is Associate Professor, Division of Biostatistics, University of California, Berkeley, Berkeley, CA 94720-3370 (Email: lexinli@berkeley.edu).
Hua Zhou is Assistant Professor, Department of Statistics, North Carolina State University, Raleigh, NC 27695-8203 (Email: hua\_zhou@ncsu.edu).
Dinggang Shen is Professor, Department of Radiology, University of North Carolina, Chapel Hill, NC 27599-7420 (E-mail: dinggang\_shen@med.unc.edu).
The Alzheimer's Disease Neuroimaging Initiative: Data used in the preparation of this article were obtained from the ADNI data base (\url{http://adni.loni.usc.edu/}). As such, the investigators within the ADNI contributed to the design and implementation of ADNI and/or provided data but did not participate in the analysis or writing of this article. A complete listing of ADNI investigators is available at: \url{http://www.loni.usc.edu/ADNI/Data/ADNI_Authorship_List.pdf}.
}
\end{footnotetext}

\baselineskip=21pt

\begin{abstract}
In an increasing number of neuroimaging studies, brain images, which are in the form of multidimensional arrays (tensors), have been collected on multiple subjects at multiple time points. Of scientific interest is to analyze such massive and complex longitudinal images to diagnose neurodegenerative disorders and to identify disease relevant brain regions. In this article, we treat those problems in a unifying regression framework with image predictors, and propose tensor generalized estimating equations (GEE) for longitudinal imaging analysis. The GEE approach takes into account intra-subject correlation of responses, whereas a low rank tensor decomposition of the coefficient array enables effective estimation and prediction with limited sample size. We propose an efficient estimation algorithm, study the asymptotics in both fixed $p$ and diverging $p$ regimes, and also investigate tensor GEE with regularization that is particularly useful for region selection. The efficacy of the proposed tensor GEE is demonstrated on both simulated data and a real data set from the Alzheimer's Disease Neuroimaging Initiative (ADNI).
\end{abstract}

\newpage

\section{Introduction}

Analyzing brain imaging data to study neuropsychiatric and neurodegenerative disorders is gaining increasing interest in recent years \citep[among many others]{Lazar2008,  Friston2009, Hinrichs2009, KangOmbao2012, AstonKirch2012}. There are a variety of forms, or modalities, of images obtained through different imaging technologies, including magnetic resonance imaging (MRI), functional magnetic resonance imaging (fMRI), positron emission tomography (PET), and electroencephalography (EEG), among others. Regardless of image modalities, it is of common scientific interest to use brain images to diagnose neurodegenerative disorders, to predict onset of neuropsychiatric diseases, and to identify disease relevant brain regions or activity patterns. These problems can be collectively formulated as a regression with a clinical outcome and an image predictor, whereas the image takes a unifying form of multidimensional array, also known as \emph{tensor}. 

Early imaging studies typically involved only a handful of subjects. More recently, a number of brain imaging databases are emerging with a relatively large number of study subjects \citep{ADHD-url, ADNI-url}. Meanwhile, in an increasing number of studies, images were acquired for each subject not only at the baseline, but also over multiple visits, resulting in longitudinal images. Our motivating example is a study from the Alzheimer's Disease Neuroimaging Initiative (ADNI). It consists of 88 subjects with mild cognitive impairment (MCI), which is a prodromal stage of Alzheimer's disease (AD). Each subject had MRI scans at 5 different time points: baseline, 6-month, 12-month, 18-month and 24-month. After preprocessing, each MRI image is $32 \times 32 \times 32$ dimensional. Also measured for each subject at each visit was a cognitive score, the Mini-Mental State Examination (MMSE), indicating progression of the disease. It is scientifically important to understand association between MCI/AD and the structural brain atrophy as reflected by MRI. It is equally important to use MRI images to accurately predict AD/MCI, as an accurate diagnosis is critical for timely therapy and possible delay of the disease \citep{ZhangShen2011}. 

While there has been an enormous literature on imaging analysis for AD, most existing methods perform the prediction using only the baseline data, ignoring data at the follow-up time points that often contain useful longitudinal information. Recently, a small group of researchers started to use longitudinal imaging data for individual-based classification \citep{Misra09LongtitudinalPrediction,Davatzikos09Longitudinal,McEvoy11Longitudinal,Hinrichs11Predict} and for cognitive score prediction \citep{ZhangShen2012}, whereas a limited number of studies regressed longitudinal image responses on a collection of covariates, first one voxel at a time then spatially smoothing the parameters \citep{Skup2012,LiShen2013}. In general, longitudinal imaging analysis is challenging, due to both the ultrahigh dimensionality and the complex spatial structure of images, while the longitudinal correlation adds another layer of complication. 

Since the seminal work of \citet{LiangZeger1986}, there has been a substantive literature on statistical analysis of longitudinal data. See \citet{Prentice1991, LiB1997, Qu2000, Xie2003, Balan2005, SongQu2009, Wang2011}, among many others. There is also a line of research studying variable selection for longitudinal models, including \citet{Pan2001, FanLi2004, Ni2010, Xue2010, Wang2012}. However, all those studies take the covariates as a vector, whereas in imaging regression, covariates take the form of multi-dimensional arrays. Naively turning an array into a vector would result in extremely high dimensionality. For instance, a $32 \times 32 \times 32$ MRI image would require $32^3 = 32,768$ parameters. Moreover, vectorization destroys inherent spatial information in images. There have been some recent developments of statistical regression models for image/tensor covariates; for instance, \citet{Caffo2010,ReissOgden2010,Wang2014}. In particular, \citet{ZhouLiZhu2013} proposed a class of tensor regression models by imposing a low rank tensor decomposition on the coefficient tensor. Although those methods directly work with a tensor covariate, none has taken longitudinal tensors into account, and thus none is immediately applicable to our longitudinal imaging study.

In this article, we propose tensor generalized estimating equations for longitudinal imaging analysis. Our proposal consists of two key components: a low rank tensor decomposition and generalized estimating equations (GEE). Similar to \citet{ZhouLiZhu2013}, we choose to impose a low rank structure, the CANDECOMP/PARAFAC (CP) decomposition \citep{KoldaBader09Tensor}, on the coefficient array in GEE. This structure substantially reduces the number of free parameters and makes subsequent estimation and inference feasible. But unlike \citet{ZhouLiZhu2013}, we incorporate this low rank structure in estimating equations to accommodate longitudinal correlation of the data. We have chosen GEE over another popular approach, the mixed effects model, for longitudinal imaging analysis. This is because the GEE approach only requires the first two marginal moments and a working correlation structure for the scalar response variable. By contrast, a mixed effects model requires specification of a distribution for the parameters, which turns out to be a formidable task for a tensor covariate. Within the tensor GEE framework, we develop a scalable computation algorithm for solving the complicated tensor estimating equations. Next we establish the asymptotic properties of the solution of tensor GEE, including consistency and asymptotic normality under two large sample scenarios: the number of parameters is fixed and the number of parameters diverges along with the sample size. In particular, we show that the tensor GEE estimator inherits the robustness feature of the classical GEE estimator, in that the estimate is consistent even if the working correlation structure is misspecified. Finally, we investigate regularization in the context of tensor GEE. Regularization is crucial when the number of parameters far exceeds the sample size, and is also useful for stabilizing estimates and incorporating prior subject knowledge. For instance, employing an $L_1$ penalty in our tensor GEE in effect finds subregions of brains that are highly relevant to the clinical outcome. This region selection is of scientific interest itself, and corresponds to the intensively studied variable selection problem in classical regressions with vector-valued predictors.

Our contributions are two-fold. First of all, our proposal offers a timely response to the increasing availability of longitudinal imaging data along with the growing interest of their analysis. To the best of our knowledge, there has been very few systematic statistical methods developed for such an analysis. Second, our work generalizes both the GEE approach from vector-valued covariates to tensor-valued image covariate, as well as the tensor regression model of \citet{ZhouLiZhu2013} from independent imaging data to longitudinal imaging data. Such a generalization parallels the extension in classical regressions with vector predictors. This extension, however, is far from trivial. Owing to the intrinsic complexity of both spatially and temporally correlated observations as well as the huge data size, longitudinal imaging analysis is much more challenging than both longitudinal analysis with vector-valued predictors and imaging analysis at a single time point. Given that the results of this kind are rare, our proposal offers a useful addition to the literature of both longitudinal and imaging analysis. 

The rest of the article is organized as follows. Section \ref{sec:tgee} proposes tensor GEE for longitudinal imaging data, along with their estimation and regularization. Section \ref{sec:theory} presents the asymptotic results for the tensor GEE estimates. Simulation studies and real data analysis are carried out in Sections \ref{sec:simulations} and \ref{sec:realdata}, respectively, followed by a discussion in Section \ref{sec:discussion}.

\section{Tensor Generalized Estimating Equations}
\label{sec:tgee}

\subsection{Notations and Preliminaries}
\label{sec:notation}

Suppose there are $n$ training subjects, and for the $i$-th subject, there are observations over $m_i$ time points. For simplicity, we assume $m_i=m$ and the time points are the same for all subjects. The observed data consist of $\{(Y_{ij}, \Xbf_{ij}, \Zbf_{ij}), i=1, \ldots, n, j=1, \ldots, m \}$, where, for the $i$-th subject at the $j$-th time point, $Y_{ij}$ denotes the target response, $\Zbf_{ij}\in \real{p_0}$ is a conventional predictor vector, and $\Xbf_{ij} \in \real{p_1\times \dots \times p_D}$ is a $D$-dimensional array that represents the image covariate. The array dimension $D$ depends on the image modality. With an image at a single time point, for EEG, $D=2$, for MRI and PET, $D=3$, and for fMRI, $D=4$. Write $\Ybf_i = (Y_{i1}, \dots, Y_{im})\trans$. A key attribute of longitudinal data is that the observations from different subjects are commonly assumed independent, but the observations from the same subject are \emph{correlated}. That is, the intra-subject covariance matrix, $\textrm{Var}(\Ybf_i) \in \real{m \times m}$, is not a diagonal matrix but with some structure. 

Next we review some key notations and operations of multidimensional array that will be used throughout this article. The \emph{inner product} between two tensors is defined as $\langle \Bbf,\Xbf \rangle = \langle \vect \Bbf, \vect \Xbf \rangle = \sum_{i_1,\ldots,i_D} \beta_{i_1\ldots i_D} x_{i_1 \ldots i_D}$, where the \emph{$\vect(\Bbf)$ operator} stacks the entries of a tensor $\Bbf \in \real{p_1 \times \cdots \times p_D}$ into a column vector. The \emph{outer product}, $\bbf_1 \circ \bbf_2 \circ \cdots \circ \bbf_D$, of $D$ vectors $\bbf_d \in \real{p_d}$ is a $p_1 \times \cdots \times p_D$ array with entries $(\bbf_1 \circ \bbf_2 \circ \cdots \circ \bbf_D)_{i_1 \cdots i_D} = \prod_{d=1}^D b_{di_d}$. The \emph{mode-$d$ matricization}, $\Bbf_{(d)}$, flattens a tensor $\Bbf$ into a $p_d \times \prod_{d' \ne d} p_{d'}$ matrix such that the $(i_1,\ldots,i_D)$ element of the array $\Bbf$ maps to the $(i_d,j)$ element of the matrix $\Bbf_{(d)}$, where $j = 1 + \sum_{d'\ne d} (i_{d'}-1) \prod_{d''<d',d'' \ne d} p_{d''}$.

A tensor $\Bbf \in \real{p_1 \times \cdots \times p_D}$ is said to admit a \emph{rank-$R$ CANDECOMP/PARAFAC (CP) decomposition} \citep{KoldaBader09Tensor}, if
\begin{eqnarray}
\Bbf = \sum_{r=1}^R \betabf_1^{(r)} \circ \cdots \circ \betabf_D^{(r)},   \label{eqn:R-CP-decomp}
\end{eqnarray}
where $\betabf_d^{(r)} \in \real{p_d}, d=1,\ldots,D,r=1,\ldots,R$, are all column vectors, and $\Bbf$ cannot be written as a sum of less than $R$ outer products. The decomposition (\ref{eqn:R-CP-decomp}) is often represented by a shorthand, $\Bbf = \llbracket \Bbf_1,\ldots,\Bbf_D \rrbracket$,  where $\Bbf_d = [\betabf_d^{(1)}, \ldots, \betabf_d^{(R)}] \in \real{p_d \times R}$. If a tensor $\Bbf \in \real{p_1 \times \cdots \times p_D}$ admits a rank-$R$  decomposition (\ref{eqn:R-CP-decomp}), then
\begin{eqnarray*}
\Bbf_{(d)} &= \Bbf_d (\Bbf_D \odot \cdots \odot \Bbf_{d+1} \odot \Bbf_{d-1} \odot \cdots \odot \Bbf_1) \trans  \, \text{ and}~~
\vect \, \Bbf    =   (\Bbf_D \odot \cdots \odot \Bbf_1) {\bf 1}_{R},
\end{eqnarray*}
where $\odot$ denotes the \emph{Khatri-Rao product} \citep{RaoMitra71GenInv} of two matrices $\Bbf_{d} \in \real{p_{d} \times r}$ and $\Bbf_{d'} \in \real{p_{d'} \times r}$ such that $\Bbf_{d} \odot \Bbf_{d'} = \left[\betabf_{d}^{(1)} \otimes \betabf_{d'}^{(1)} \; \betabf_{d}^{(2)} \otimes \betabf_{d'}^{(2)} \; \ldots \; \betabf_{d}^{(R)} \otimes \betabf_{d'}^{(R)} \right] \in \real{p_d p_{d'} \times R}$, and $\otimes $ denotes the \emph{Kronecker product}.

\subsection{Tensor Generalized Estimating Equations}
\label{sec:tensor-gee}

The GEE method has been widely employed for analyzing correlated longitudinal data since the pioneer work of \citet{LiangZeger1986}. It requires specification of the first two moments of the conditional distribution of the response given the covariates, $\mu_{ij} = E(Y_{ij} | \Xbf_{ij}, \Zbf_{ij})$ and $\sigma^2_{ij}=\text{Var}(Y_{ij} | \Xbf_{ij}, \Zbf_{ij})$. Following \citet{LiangZeger1986}, we assume $Y_{ij}$ is from an exponential family with canonical link. Then
\begin{eqnarray*}
\mu_{ij}(\Bbf, \gammabf) = \mu(\theta_{ij}), \; \textrm{ and } \;
\sigma^2_{ij}(\Bbf, \gammabf) = \phi \mu^{(1)}(\theta_{ij}), \quad i=1, \ldots, n, \; j=1, \ldots, m, 
\end{eqnarray*}
where $\mu(\cdot)$  is a differentiable canonical link function, $\mu^{(1)}(\cdot)$ is its first derivative, $\theta_{ij}$ is the linear systematic part, and $\phi$ is an over-dispersion parameter. In this article we simply set $\phi=1$ while the extension to a general $\phi$ is straightforward. $\theta_{ij}$ is associated with the covariates via the relation
\begin{eqnarray} \label{eqn:systematic-part-general}
\theta_{ij} = \gammabf\trans \Zbf_{ij} + \langle \Bbf, \Xbf_{ij} \rangle,
\end{eqnarray}
where $\gammabf$ is the coefficient vector associated with the covariate vector $\Zbf$, including the intercept, and $\Bbf$ is the coefficient tensor of the same size as $\Xbf$ that captures effects of every array element of $\Xbf$. 

The GEE estimator of $\Bbf, \gammabf$ is then defined as the solution of
\begin{eqnarray} \label{eqn:gee}
\sum_{i=1}^n \left\{ \frac{\partial \mubf_i(\Bbf, \gammabf)}{\partial [\vect(\Bbf)\trans, \gammabf\trans]\trans} \right\}\trans \Vbf_i^{-1} \bigl\{ \Ybf_i - \mubf_i(\Bbf, \gammabf) \bigr\} = \zerobf,
\end{eqnarray}
where $\Ybf_i = (Y_{i1}, \dots, Y_{im})\trans$, $\mubf_i(\Bbf, \gammabf) = [\mu_{i1}(\Bbf, \gammabf), \ldots, \mu_{im}(\Bbf, \gammabf)]\trans$, and $\Vbf_i=\text{cov}(\Ybf_i)$ is the response covariance matrix of the $i$-th subject.  The first component in (\ref{eqn:gee}) is the derivative of $\mubf_i(\Bbf, \gammabf)$ with respect to the vector $[\vect(\Bbf)\trans, \gammabf\trans]\trans \in \real{p_0 + \prod_d p_d}$. As such, there are totally $p_0 + \prod_d p_d$ estimating equations to solve in (\ref{eqn:gee}). For regression with image covariates, this dimension is ultrahigh and usually far exceeds the sample size. For instance, for a regression with a $32 \times 32 \times 32$ MRI image predictor, an intercept, and two additional scalar covariates, the number of equations to solve is in the scale of $32^3 + 3 = 327, 71$, resulting no unique solution when the sample size is only in hundreds. It thus becomes crucial to reduce the number of estimating equations. 

Toward that end, we impose a low rank structure on the coefficient array $\Bbf$. More specifically, we assume $\Bbf$ in model (\ref{eqn:systematic-part-general}) follows a CP structure in (\ref{eqn:R-CP-decomp}), $\Bbf = \llbracket \Bbf_{1},\ldots,\Bbf_{D} \rrbracket$,  where $\Bbf_d = [\betabf_d^{(1)}, \ldots, \betabf_d^{(R)}] \in \real{p_d \times R}$. Then the systematic part in (\ref{eqn:systematic-part-general}) becomes
\begin{eqnarray}
\theta_{ij}  & = & \gammabf \trans \Zbf_{ij} +  \langle \sum_{r=1}^R \betabf_1^{(r)} \circ \cdots \circ \betabf_D^{(r)},\Xbf_{ij}  \rangle  \nonumber \\
& = & \gammabf \trans \Zbf_{ij} + \langle (\Bbf_D \odot \cdots \odot \Bbf_1) {\bf 1}_{R} , \vect \Xbf_{ij} \rangle.
\label{eqn:r-tensorreg}
\end{eqnarray}
Adopting (\ref{eqn:r-tensorreg}), we propose the \emph{tensor generalized estimating equations} estimator of $\Bbf, \gammabf$, defined as the solution of 
\begin{eqnarray} \label{eqn:tgee}
\sum_{i=1}^n \left\{ \frac{\partial \mubf_i(\Bbf, \gammabf)}{\partial [\betabf_\Bbf\trans, \gammabf\trans]\trans} \right\}\trans \Vbf_i^{-1} \bigl\{ \Ybf_i - \mubf_i(\Bbf, \gammabf) \bigr\} = \zerobf,
\end{eqnarray}
where $\betabf_\Bbf = \vect(\Bbf_1, \ldots, \Bbf_D)$, and the subscript {\scriptsize $\Bbf$} is to remind that $\betabf$ is constructed based on the CP decomposition of a given coefficient tensor $\Bbf = \llbracket \Bbf_1,\ldots,\Bbf_D \rrbracket$. Comparing to the classical GEE (\ref{eqn:gee}), the derivative is now with respect to $[\betabf_\Bbf\trans, \gammabf\trans]\trans \in \real{p_0 + R \sum_d p_d}$. Consequently, the number of estimating equations has reduced from the exponential order $p_0 + \prod_d p_d$ to the linear order $p_0 + R \sum_d p_d$. This substantial reduction in dimensionality, as we will demonstrate later, enables effective estimation and inference, and also provides a sound recovery of both low rank and high rank signals.

Examining (\ref{eqn:tgee}), the true intra-subject covariance structure $\Vbf_i$ is usually unknown in practice. The classical GEE adopts a working covariance matrix, specified through a working correlation matrix $\Rbf$. That is, $\Vbf_i = \Abf^{1/2}_i(\Bbf, \gammabf) \Rbf \Abf^{1/2}_i(\Bbf, \gammabf)$, where $\Abf_i(\Bbf, \gammabf)$ is an $m \times m$ diagonal matrix with $\sigma^2_{ij}(\Bbf,\gammabf)$ on the diagonal and $\Rbf$ is the $m$-by-$m$ working intra-subject correlation matrix. Some commonly used correlation structures include independence, autocorrelation (AR), compound symmetry, and unstructured correlation, among others. The correlation matrix $\Rbf$ may involve additional parameters, which can be estimated using residual-based moment method. 

By both adopting this working covariance/correlation idea, and explicitly evaluating the derivative in (\ref{eqn:tgee}), we finally arrive at the formal definition of the tensor GEE estimator, which is the solution $(\widehat \Bbf, \hat \gammabf)$ of the following estimating equations
\begin{equation} \label{eqn:tgee-corr}
\sum_{i=1}^n
\begin{pmatrix}
[\Jbf_1 \ldots \Jbf_D]\trans \vect(\Xbf_i) \\
(\Zbf_{i1}, \ldots, \Zbf_{im})
\end{pmatrix}
\Abf^{1/2}_i(\Bbf, \gammabf){\bf \widehat{R}}^{-1}{\Abf}^{-1/2}_i(\Bbf, \gammabf) \bigl\{ \Ybf_i - \mubf_i(\Bbf, \gammabf) \bigr\} = \zerobf,
\end{equation}
where $\widehat{\Rbf}$ is an estimated correlation matrix, $\vect(\Xbf_i) = (\vect(\Xbf_{i1}), \ldots, \vect(\Xbf_{im}))$ is a $\prod_{d=1}^D p_d \times m$ matrix, $\Jbf_d$ is the $\prod_{d=1}^D p_d \times  p_d R$ Jacobian matrix of the form $\Pibf_d [(\Bbf_D \odot \cdots \odot \Bbf_{d+1} \odot \Bbf_{d-1} \odot \cdots \odot \Bbf_1) \otimes \Ibf_{p_d}]$, where $\Pibf_d$ is the $(\prod_{d=1}^D p_d)$-by-$(\prod_{d=1}^D p_d)$ permutation matrix that reorders $\vect \Bbf_{(d)}$ to obtain $\vect \Bbf$, i.e., $
    \vect \Bbf = \Pibf_d \, \vect \Bbf_{(d)}$. Note that $\mu^{(1)}(\theta_{ij})$ has been canceled by the diagonal of the matrix $\Abf^{-1}_i$ due to the property of canonical link. For ease of presentation, we denote the left hand side of equation (\ref{eqn:tgee-corr}) as $\sbf(\Bbf, \gammabf)$, and write the tensor GEE (\ref{eqn:tgee-corr}) as $\sbf(\Bbf, \gammabf) = \zerobf$.

\subsection{Estimation}

Directly solving the tensor generalized estimating equations (\ref{eqn:tgee-corr}) with respect to $(\Bbf, \gammabf)$ can be computational intensive, as the mean function of the response given the covariates is nonlinear in the parameters and the Jacobian matrices $\Jbf_1, \ldots, \Jbf_D$ also depend on the unknown parameters. We propose to iteratively solve the sub-GEE for $\Bbf_1, \ldots, \Bbf_D$, along with $\gammabf$, one at a time, while keeping all other components fixed. When updating $\Bbf_d \in \real{p_d \times R}$, the systematic part $\theta_{ij}(\Bbf,\gammabf)$ can be rewritten as
\begin{eqnarray*}
\theta_{ij}(\Bbf,\gammabf) & = & \gammabf\trans \Zbf_{ij} + \langle \Bbf, \Xbf_{ij} \rangle \\ 
& = & \gammabf\trans \Zbf_{ij} + 
\langle \Bbf_d, \Xbf_{ij(d)} (\Bbf_D \odot \cdots \odot \Bbf_{d+1} \odot \Bbf_{d-1} \odot \cdots \odot \Bbf_1) \rangle,
\end{eqnarray*}
where $\Xbf_{ij(d)}$ is the mode-$d$ matricization of the tensor $\Xbf_{ij}$. As such, the systematic part $\theta_{ij}(\Bbf, \gammabf)$ becomes linear in $\Bbf_d$. The Jacobian matrix $\Jbf_d$ is free of $\Bbf_d$ and depends on the covariates and fixed parameters only. Consequently, each step reduces to a standard GEE problem with $Rp_d$ parameters, which can be solved using standard statistical softwares.

A problem of practical interest is to choose the rank $R$ for $\Bbf$ in its CP decomposition. This can be viewed as a \emph{model selection} problem. \citet{Pan2001} proposed a quasi-likelihood independence model criterion for the classical GEE model selection, by evaluating the likelihood under the independence working correlation assumption. In our tensor GEE setup,  we use the following BIC-type information criterion
\begin{equation} \label{eqn:bic}
\text{BIC}(R)=-2 \ell(\widehat \Bbf(R), \hat \gammabf; \Ibf_m)+\log(n)p_e,
\end{equation}
where $\ell(\widehat \Bbf(R), \hat \gammabf; \Ibf_m)$ is the log-likelihood evaluated at the tensor GEE estimator $\hat \gammabf$ and $\widehat \Bbf(R)$ with a working rank $R$ and the independence working correlation structure $\Ibf_m$. For simplicity, we call this criterion BIC, as the term $\log(n)$ is used. Because the CP decomposition itself is not unique, but can be made so under some minor conditions \citep{ZhouLiZhu2013}, the actual number of estimating equations, or the effective number of parameters, is of the form: $p_e = R(p_1+p_2)-R^2$ for $D=2$, and $p_e = R(\sum_d p_d - D+1)$ for $D>2$. We choose $R$ that minimizes this criterion among a series of working ranks. We will briefly illustrate its use in Section~\ref{sec:signal}.

\subsection{Regularization}
\label{sec:regularization}

Even after introducing a low rank structure in our tensor GEE, regularization can still be useful, as the number of subjects is often limited in a neuroimaging study. In this section, we consider a general form of regularized tensor GEE that includes a variety of penalty functions. Then in Section~\ref{sec:lasso}, we will illustrate with a lasso penalty that is capable of identifying sub-regions of brains associated with the clinical outcome. Specifically, we consider the following regularized tensor GEE
\begin{align*}
    \sbf(\Bbf,\gammabf) + \begin{pmatrix} \partial_{\beta_{11}^{(1)}} P_\lambda(|\beta_{11}^{(1)}|,\rho) \\ \vdots \\ \partial_{\beta_{di}^{(r)}} P_\lambda(|\beta_{di}^{(r)}|,\rho) \\ \vdots \\ \partial_{\beta_{Dp_D}^{(R)}} P_\lambda(|\beta_{Dp_D}^{(R)}|,\rho) \end{pmatrix} = \zerobf_{p_e}, 
\end{align*}
where $P_\lambda(|\beta|,\rho)$ is a scalar penalty function, $\rho$ is the penalty tuning parameter,  $\lambda$ is an index for the penalty family, $\partial_\beta P_\lambda(|\beta|,\rho)$ is the subgradient with respect to argument $\beta$, and the subscript $p_e$ of $\zerobf$ is a reminder of the number of estimating equations to solve. Some widely used penalties include: power family \citep{FrankFriedman93Bridge},  in which $P_\lambda(|\beta|,\rho) = \rho |\beta|^\lambda$, $\lambda \in (0,2]$, and in particular lasso \citep{Tibshirani96Lasso} ($\lambda=1$) and ridge ($\lambda=2$); elastic net \citep{ZouHastie05Enet}, in which $P_\lambda(|\beta|, \rho) = \rho [(\lambda-1) \beta^2/2 + (2-\lambda) |\beta|], \lambda \in [1,2]$; and SCAD \citep{FanLi01SCAD}, in which $\partial / \partial |\beta| P_{\lambda}(|\beta|, \rho) = \rho \left\{ 1_{\{|\beta| \le \rho\}} + (\lambda \rho - |\beta|)_+ /(\lambda-1)\rho 1_{\{|\beta| > \rho\}} \right\}$, $\lambda>2$, among many others.

Thanks to the separability of parameters in the regularization term, the alternating updating strategy still applies. When updating $\Bbf_d$, we solve the penalized sub-GEE
\begin{eqnarray}
    \sbf_d(\Bbf_d) + \begin{pmatrix} \partial_{\betabf_{d1}^{(1)}} P_\lambda(|\beta_{d1}^{(1)}|,\rho)  \\ \vdots \\ \partial_{\betabf_{di}^{(r)}} P_\lambda(|\beta_{di}^{(r)}|,\rho) \\ \vdots \\ \partial_{\betabf_{dp_d}^{(R)}} P_\lambda(|\beta_{dp_D}^{(R)}|,\rho)  \end{pmatrix} = \zerobf_{R p_d},	\label{eqn:sub-pen-GEE}
\end{eqnarray}
where $\sbf_d$ is the sub-estimation equation for block $\Bbf_d$, and there are $R p_d$ equations to solve at this step. Anti-derivative of $\sbf_d$ is recognized as the loss of an Aitken linear model with block diagonal covariance matrix. Thus after linear transformation of $\Ybf_i$ and the working design matrix, solution to \eqref{eqn:sub-pen-GEE} is same as the minimizer of a regular penalized weighted least squares problem, for which many software packages exist. The fitting procedure boils down to alternating penalized weighted least squares problem.

\section{Theory}
\label{sec:theory}

In this section, we study the asymptotic properties of the unregularized tensor GEE estimator as the number of subjects $n$ goes to infinity, while we assume the true rank of the tensor coefficient is known. We investigate two scenarios: the number of parameters is fixed in Section~\ref{sec:asy-fixed}, and the number of parameters diverges in Section~\ref{sec:asy-diverge}. For ease of exposition, we omit the vector-valued covariates $\Zbf$ and the associated parameters $\gammabf$, while the results can be easily extended to incorporate them. Our development builds upon and extends the previous work of \citet{Xie2003, Balan2005, Wang2011} from classical vector GEE to tensor GEE, while we spell out the similarity as well as difference in asymptotics when comparing the vector and tensor GEE. We show that tensor GEE estimator inherits the key advantage of the classical GEE estimator in that it remains consistent even if the working correlation structure is misspecified. On the other hand, we note that, although one can generalize the classical GEE asymptotics by directly vectorizing the tensor, it would have to require a more stringent set of conditions. By contrast, we could achieve the robustness in consistency for our tensor GEE based on a weaker set of conditions, and we achieve this by imposing and exploiting the special structure of the coefficient tensor.

\subsection{Asymptotics for Fixed Dimension}
\label{sec:asy-fixed}

We begin with the list of regularity conditions for the asymptotics of tensor GEE with a fixed number of parameters.  

\begin{enumerate}

\item[(A1)] The elements of $\Xbf_{ij}$, $i=1, \dots, n$, $j=1, \dots, m$, are uniformly bounded by a finite constant. 

\item[(A2)] The true value $\Bbf_0$  of the unknown parameter lies in the interior of a compact parameter space $\mathcal{B}$ and follows a rank-$R$ CP structure defined in (\ref{eqn:R-CP-decomp}). 

\item[(A3)] Letting $I(\Bbf)=n^{-1}\sum_{i=1}^n[\Jbf_1 \ldots \Jbf_D]\trans \vect(\Xbf_i) \vect(\Xbf_i)\trans [\Jbf_1 \ldots \Jbf_D]$. It is assumed that there exist two positive constants $c_1<c_2$ such that 
\begin{equation*}
c_1 \le \lambda_{\min}(I(\Bbf)) \le \lambda_{\max}(I(\Bbf)) \le c_2,
\end{equation*} 
over the set $\{\Bbf: ||\betabf_{\Bbf}-\betabf_{\Bbf_0}||\le \triangle n^{-1/2}\}$ for some constant $\triangle>0$, where $\lambda_{\min}$ and $\lambda_{\max}$ are smallest and largest eigenvalue, respectively. It is also assumed that on the same set $I(\Bbf)$ has a constant rank. 

\item[(A4)] The true intra-subject correlation matrix $\Rbf_0$ has bounded eigenvalues from zero and infinity. The estimated working correlation matrix satisfies $\| \widehat{\Rbf}^{-1} - \tilde{\Rbf}^{-1} \|_\textbf{F} = O_p(n^{-1/2})$, where $\|\cdot\|_\textbf{F}$ is the Frobenius norm,  $\tilde{\Rbf}$ is some positive definite matrix with bounded eigenvalues from zero and infinity, and $\tilde{\Rbf} = \Rbf_0$ is \emph{not} required.

\item[(A5)] For some constant $\delta>0$ and $M_1>0$, $E(\| \Abf^{-1/2}_i(\Bbf_0)(\Ybf_i - \mubf_i(\Bbf_0)) \|)^{2+\delta} \le M_1$ for all $1 \le i \le n$, where $\Abf^{-1/2}_i(\Bbf_0)$ is the covariance matrix of $\Ybf_i $.

\item[(A6)] $\sigma^{-1}_{ij}(\Bbf_0)(Y_{ij}-\mu_{ij}(\Bbf_0))$ has sub-Gaussian tails for all $i=1, \dots, n$, $j=1, \dots, m$. 

\item[(A7)] The elements of $\partial \theta_{ij}(\betabf_{\Bbf_0})/\partial \betabf_{\Bbf_0}$, $i=1, \dots, n$, $j=1, \dots, m$, are uniformly bounded by a finite constant.

\item[(A8)] Denote $\mu^{(k)}(\theta_{ij})$ the $k$-th derivative of $\mu(\theta_{ij})$, where $\theta_{ij}$ is the linear systematic part evaluated at the GEE solution $\widehat \Bbf$. It is assumed that $\mu^{(1)}(\theta_{ij})$ are uniformly bounded away from zero and infinity, and $\mu^{(k)}(\theta_{ij})$ are uniformly bounded by a finite constant, over the set $\{\Bbf: ||\betabf_{\Bbf}-\betabf_{\Bbf_0}||\le \triangle n^{-1/2}\}$, for some constant $\triangle>0$, $i=1, \dots, n$, $j=1, \dots, m$, and $k=2, 3$.

\item[(A9)] Denote $ \Hbf(\Bbf,\Xbf_{ij})=\partial [\Jbf_1 \cdots \Jbf_D]\trans \text{vec}(\Xbf_{ij})/\partial \text{vec}\trans(\Bbf)$. $\Hbf(\Bbf,\Xbf_{ij})$ is the Hessian of the linear systematic part $\theta_{ij}$ under tensor structure. There exist two positive constants $c_3<c_4$ such that  \begin{equation*}
c_3 \le \lambda_{\min}(\Hbf(\Bbf,\Xbf_{ij})) \le \lambda_{\max}(\Hbf(\Bbf,\Xbf_{ij})) \le c_4,
\end{equation*} 
over the set $\{\Bbf: ||\betabf_{\Bbf}-\betabf_{\Bbf_0}||\le \triangle n^{-1/2}\}$ for some constant $\triangle>0$, $i=1, \dots, n$ and $j=1, \dots, m$.

\end{enumerate}

\noindent
A few remarks are in order. Conditions (A2) and (A3) are required for model identifiability of tensor GEE \citep{ZhouLiZhu2013}. We observe that, the matrix $I(\Bbf)$ in (A3) is an $R\sum_{d=1}^Dp_d \times R\sum_{d=1}^Dp_d$ matrix, and thus (A3) is much weaker than the nonsingularity condition on the design matrix if one were to directly vectorize the tensor covariate. Condition (A4) is commonly imposed in the GEE literature. It only requires a consistent estimator $\widehat{\Rbf}$ of some $\tilde \Rbf$, in the sense $\| \widehat{\Rbf}^{-1} - \tilde{\Rbf}^{-1} \|_\textbf{F} = O_p(n^{-1/2})$. $\tilde \Rbf$ needs to be well behaved in that it is positive definite with bounded eigenvalues from zero and infinity, but $\tilde \Rbf$ does \emph{not} have to be the true intra-subject correlation $\Rbf$. This condition essentially leads to the robust feature in Theorem~\ref{thm:consistency} that the tensor GEE estimate is consistent even if the working correlation structure is misspecified. Conditions (A5) and (A6) regulate the tail behavior of the residuals so that the noise cannot accumulate too fast, and we can employ the Lindeberg-Feller central limit theorem to control the asymptotic behavior of the residuals. Condition (A7) states the gradients of the systematic part evaluated at the truth are well-defined. Condition (A8) concerns the canonical link and generally holds for common exponential families, for example, the binomial distribution with $\mu(\theta_{ij})=\exp{\theta_{ij}}/(1+\exp{\theta_{ij}})$, and the Poisson distribution with $\mu(\theta_{ij})=\exp{\theta_{ij}}$. Condition (A9) ensures that the Hessian matrix of the linear systematic part, which is highly sparse, is well-behaved in a neighborhood of the true value. 

Before we turn to the asymptotics of the tensor GEE estimator, we address two components involved in the estimating equations: the initial estimator and the correlation estimator. Recall the tensor GEE estimator $\widehat \Bbf$ is obtained by solving the equations in (\ref{eqn:tgee-corr}). After dropping the covariate vector $\Zbf$, the tensor estimating equations become 
\begin{equation}  \label{eqn:tgee-corr-noz}
\sum_{i=1}^n[\Jbf_1 \ldots \Jbf_D]\trans \vect(\Xbf_i) \Abf^{1/2}_i(\Bbf) \widehat{\Rbf}^{-1} \Abf^{-1/2}_i(\Bbf) \bigl\{ \Ybf_i - \mubf_i(\Bbf) \bigr\} = \zerobf,
\end{equation}
where $\widehat{\Rbf}$ is any estimator of the intra-subject correlation matrix satisfying the condition (A4). We still denote the left hand side by $\sbf(\Bbf)$. Note that (\ref{eqn:tgee-corr-noz}) involves the unknown correlation $\Rbf$, and its estimate $\widehat{\Rbf}$ is often obtained via residual-based moment method, which in turn requires an initial estimator of $\Bbf$. Next, we examine some frequently used estimators of $\widehat{\Bbf}$ and $\widehat{\Rbf}$. 

A customary initial estimator $\widehat{\Bbf}$ in the GEE literature is the one that assumes an independent working correlation. That is, one completely ignores possible intra-subject correlation, and the corresponding tensor GEE becomes 
\begin{equation*} 
\sum_{i=1}^n[\Jbf_1 \ldots \Jbf_D]\trans \vect(\Xbf_i) \bigl\{ \Ybf_i - \mubf_i(\Bbf) \bigr\} = \zerobf.
\end{equation*}
Denoting the equations as $\sbf_{init}(\Bbf) = \zerobf$, and the solution as $\widehat \Bbf_{init}$, the next Lemma shows that it is a consistent estimator of the true $\Bbf_0$. 

\begin{lemma} \label{thm:init-est}
Under conditions (A1)-(A3) and (A5)-(A9), there exists a root $\widehat \Bbf_{init}$ of the equations $\sbf_{init}(\Bbf) = \zerobf$ satisfing that 
\begin{eqnarray*}
\| \betabf_{\widehat{\Bbf}_{init}} - \betabf_{\Bbf_0} \| = O_p(n^{-1/2}). 
\end{eqnarray*}
\end{lemma}
\noindent
Here $\betabf_\Bbf = \vect(\Bbf_1, \ldots, \Bbf_D)$, and is constructed based on the CP decomposition of a given tensor $\Bbf = \llbracket \Bbf_1,\ldots,\Bbf_D \rrbracket$, as defined before. 

Given a consistent initial estimator of $\Bbf_0$, there exist multiple choices for the working correlation structure, e.g., autocorrelation, compound symmetry, and the nonparametric structure \citep{Balan2005}. We will investigate those choices in our simulations and real data analysis. 

Next we establish the consistency and asymptotic normality of the tensor GEE estimator from (\ref{eqn:tgee-corr-noz}).

\begin{theorem} \label{thm:consistency}
Under conditions (A1)-(A9), there exists a root $\widehat \Bbf$ of the equations $\sbf(\Bbf) = \zerobf$ satisfing that 
\vspace{-0.1in}
\begin{eqnarray*}
\| \betabf_{\widehat{\Bbf}} - \betabf_{\Bbf_0} \| = O_p(n^{-1/2}). 
\end{eqnarray*}
\end{theorem}

\noindent
The key message of Theorem \ref{thm:consistency}, as implied by condition (A4), is that the consistency of the tensor coefficient estimator $\widehat \Bbf$ does \emph{not} require the estimated working correlation $\widehat{\Rbf}$ being a consistent estimator of the true correlation $\Rbf$. This protects us from potential misspecification of the intra-subject correlation structure. Such a robustness feature is well known for GEE estimator with vector-valued covariates. Theorem \ref{thm:consistency} confirms and extends this result to the tensor GEE case with image covariates. We also remark that, although the asymptotics of the classical GEE can in principle be generalized to the tensor data by directly vectorizing the coefficient array, the ultrahigh dimensionality of the parameters would have made the regularity conditions such as (A3) unrealistic. By contrast, Theorem \ref{thm:consistency} ensures that one could still enjoy the consistency and robustness properties, by taking into account the structural information of the tensor coefficient under the GEE framework.

Under condition (A4), we define 
\begin{eqnarray*}
\tilde\Mbf_n(\Bbf) & = & \sum_{i=1}^n [\Jbf_1 \ldots \Jbf_D]\trans \vect(\Xbf_i)\Abf^{1/2}_i(\Bbf) \tilde{\Rbf}^{-1} \Rbf_0 \tilde{\Rbf}^{-1} \Abf^{1/2}_i(\Bbf) \vect\trans(\Xbf_i) [\Jbf_1 \ldots \Jbf_D], \\
\tilde\Dbf_{n1}(\Bbf) & = & \sum_{i=1}^n [\Jbf_1 \ldots \Jbf_D]\trans \vect(\Xbf_i)\Abf^{1/2}_i(\Bbf) \tilde{\Rbf}^{-1} \Abf^{1/2}_i(\Bbf) \vect\trans(\Xbf_i) [\Jbf_1 \ldots \Jbf_D]. 
\end{eqnarray*} 
As we will show in the appendix, $\tilde\Mbf_n(\Bbf)$ approximates the covariance matrix of $\sbf(\Bbf)$ in (\ref{eqn:tgee-corr-noz}), while $\tilde\Dbf_{n1}(\Bbf)$ approximates the leading term of the negative gradient of $\sbf(\Bbf)$ with respect to $\betabf_\Bbf$. Then the next theorem gives the asymptotic normality of the tensor GEE estimator. 
\begin{theorem} \label{thm:normality}
Under conditions (A1)-(A9), for any vector $\bbf \in \real{R \sum_{d=1}^D p_d}$ such that $\| \bbf \|=1$, we have
\begin{eqnarray*}
\bbf\trans \tilde{\Mbf_n}^{-1/2}(\Bbf_0) \tilde{\Dbf}_{n1}(\Bbf_0) \bigl( \betabf_{\widehat{\Bbf}} - \betabf_{\Bbf_0} \bigr) \to \mathrm{Normal} (0, 1) \textrm{ in distribution.}
\end{eqnarray*}
\end{theorem}

By Theorem \ref{thm:normality} and Cram\'er-Wold theorem, one can derive the sandwich covariance estimator of $\textrm{Var}(\betabf_{\widehat \Bbf})$, and carry out the subsequent Wald inference. Specifically, it is easy to see that the variance of the GEE estimator can be approximated by the asymptotic variance $\tilde\Dbf_{n1}^{-1}(\Bbf_0) \tilde\Mbf_n(\Bbf_0) \tilde\Dbf_{n1}^{-1}(\Bbf_0)$. Since it involves the unknown terms $\Bbf_0, \Rbf_0$ and $\tilde\Rbf$, we plug in, respectively, $\widehat{\Bbf}$, $n^{-1} \sum_{i=1}^n \Abf^{-1/2}_i(\widehat{\Bbf})\{ \Ybf_i-\mubf_i(\widehat{\Bbf}) \} \{ \Ybf_i-\mubf_i(\widehat{\Bbf}) \}\trans\Abf^{-1/2}_i(\widehat{\Bbf})$, and $\widehat{\Rbf}$, which leads to the sandwich estimator, 
\begin{eqnarray*} 
\widehat{\text{Var}}(\betabf_{\widehat{\Bbf}}) = \tilde\Dbf_{n1}^{-1}(\widehat{\Bbf}) \tilde\Mbf_n(\widehat{\Bbf}) \tilde\Dbf_{n1}^{-1}(\widehat{\Bbf}). 
\end{eqnarray*}
This sandwich formula in turn can be used to construct asymptotic confidence interval or asymptotic hypothesis testing through the usual Wald inference.

\subsection{Asymptotics for Diverging Dimension}
\label{sec:asy-diverge}

We next study the asymptotics when the number of parameters diverges. We assume that $p_d \sim p_n$ for $d=1, \dots, D$, where $a_n \sim b_n$ means $a_n=O(b_n)$ and $b_n=O(a_n)$. We also assume that the rank $R$ is fixed in the tensor GEE. Next we list the required regularity conditions. Since the conditions (A1), (A2), (A5)--(A7) are the same as in Section~\ref{sec:asy-fixed}, we only list the conditions that are different, while we relabel those same conditions as (A1$^*$), (A2$^*$), (A5$^*$)--(A7$^*$), respectively. 

\begin{enumerate}

\item[(A3$^*$)] There exist two positive constant $c_1<c_2$ such that 
\begin{equation*}
c_1 \le \lambda_{\min}(I(\Bbf)) \le \lambda_{\max}(I(\Bbf)) \le c_2,
\end{equation*} 
over the set $\{\Bbf: ||\betabf_{\Bbf}-\betabf_{\Bbf_0}||\le \triangle\sqrt{p_n/n}\}$ for some constant $\triangle>0$. It is also assumed that $I(\Bbf)$ has a constant rank on the same set. 

\item[(A4$^*$)] The true intra-subject correlation matrix $\Rbf_0$ has bounded eigenvalues from zero and infinity. The estimated working correlation matrix satisfies $\| \widehat{\Rbf}^{-1} - \tilde{\Rbf}^{-1} \|_\textbf{F} = O_p(\sqrt{p_n/n})$, where $\|\cdot\|_\textbf{F}$ is the Frobenius norm,  $\tilde{\Rbf}$ is some positive definite matrix with bounded eigenvalues from zero and infinity, and $\tilde{\Rbf} = \Rbf_0$ is not required.

\item[(A8$^*$)] It is assumed that $\mu^{(1)}(\theta_{ij})$ are uniformly bounded away from zero and infinity, and $\mu^{(k)}(\theta_{ij})$ are uniformly bounded by a finite constant, over the set $\{\Bbf: ||\betabf_{\Bbf}-\betabf_{\Bbf_0}||\le \triangle\sqrt{p_n/n}\}$, for some constant $\triangle>0$, $i=1, \dots, n$, $j=1, \dots, m$, and $k=2, 3$.

\item[(A9$^*$)] There exist two positive constants $c_3<c_4$ such that  \begin{equation*}
c_3 \le \lambda_{\min}(\Hbf(\Bbf,\Xbf_{ij})) \le \lambda_{\max}(\Hbf(\Bbf,\Xbf_{ij})) \le c_4,
\end{equation*} 
over the set $\{\Bbf: ||\betabf_{\Bbf}-\betabf_{\Bbf_0}||\le \triangle\sqrt{p_n/n}\}$ for some constant $\triangle>0$, $i=1, \dots, n$ and $j=1, \dots, m$.
\end{enumerate}

\noindent
Comparing the two sets of regularity conditions for the fixed and diverging number of parameters, the main difference is that the conditions are imposed on the set $\{\Bbf: ||\betabf_{\Bbf}-\betabf_{\Bbf_0}||\le \triangle\sqrt{p_n/n}\}$ when the number of parameters diverges. This is due to the slower convergence rate of the tensor GEE estimator with a diverging $p_n$. In addition, we note that $I(\Bbf)$ and $\Hbf(\Bbf,\Xbf_{ij})$ are no longer matrices with fixed dimensions when $p_n$ diverges. Correspondingly, we impose conditions (A3*) and (A9*) on the bounded eigenvalues, which are similar to the sparse Riesz condition for vector covariates. The latter condition has been frequently employed in the current literature of inference with diverging dimensions \citep{zhang2008sparsity, zhang2010nearly}. 

Next we present the asymptotics for the tensor GEE estimator with a diverging $p_n$. 

\begin{theorem} \label{thm:consistency_diverge}
Under conditions (A1*)-(A9*), and $p_n=o(n^{1/2})$, there exists a root $\widehat \Bbf$ of the equations $\sbf(\Bbf) = \zerobf$ satisfying that 
\begin{eqnarray*}
\| \betabf_{\widehat{\Bbf}} - \betabf_{\Bbf_0} \| = O_p(\sqrt{p_n/n}). 
\end{eqnarray*}
\end{theorem}

\noindent
It is important to note that, if one directly vectorizes the tensor covariate and applies the asymptotics of the classical GEE as in  \citet{Wang2011},  the conditions for the consistency would require $\prod_{d=1}^Dp_d=o(n^{1/2})$, i.e. $p_n=o(n^{1/(2D)})$. This rate can be much more stringent for a tensor covariate. Theorem \ref{thm:consistency_diverge}, instead, states that the consistency still holds with $p_n=o(n^{1/2})$, after imposing and exploiting the low rank tensor structure on the coefficients array. 

The asymptotic normality can also be established for a diverging $p_n$. 
\begin{theorem} \label{thm:normality_diverge}
Under conditions (A1*)-(A9*), and $p_n=o(n^{1/3})$, for any vector $\bbf_n \in \real{R \sum_{d=1}^D p_d}$ such that $\| \bbf_n\|=1$, we have
\vspace{-0.1in}
\begin{eqnarray*}
\bbf_n\trans \tilde{\Mbf}_n^{-1/2}(\Bbf_0) \tilde{\Dbf}_{n1}(\Bbf_0) \bigl( \betabf_{\widehat{\Bbf}} - \betabf_{\Bbf_0} \bigr) \to \mathrm{Normal} (0, 1) \textrm{ in distribution.}
\end{eqnarray*}
\end{theorem}

\noindent
Similarly, for the asymptotic normality to hold, the condition would have become $p_n=o(n^{1/(3D)})$ if one directly vectorizes the tensor covariate. By contrast, the tensor GEE requires $p_n=o(n^{1/3})$.

\section{Simulations}
\label{sec:simulations}

We have carried out extensive simulations to investigate the finite sample performance of our proposed tensor GEE approach. We adopt the following simulation setup. We generated the responses according to the normal model 
\begin{equation*}
\Ybf_i \sim \mathrm{MVN} (\mubf_i, \sigma^2{\bf R}_0), \quad i=1,\dots, n, 
\end{equation*}
where $\Ybf_i=(Y_{i1}, \ldots, Y_{im})\trans$, $\mubf_i=(\mu_{i1}, \ldots, \mu_{im})\trans$, $\sigma^2$ is a scale parameter, and $\Rbf_0$ is the true $m \times m$ intra-subject correlation matrix. We have chosen $\Rbf_0$ to be of an exchangeable (compound symmetric) structure with the off-diagonal coefficient $\rho = 0.8$. The mean function is of the form
\begin{equation*}
\mu_{ij}=\gammabf\trans \Zbf_{ij}+\langle \Bbf, \Xbf_{ij} \rangle, \quad i=1, \ldots, n, j=1, \ldots, m,
\end{equation*}
where $\Zbf_{ij} \in \real{5}$ denotes the covariate vector, with all elements generated from a standard normal distribution, and $\gammabf \in \real{5}$ is the corresponding coefficient vector, with all elements equal to one; $\Xbf_{ij} \in \real{64 \times 64}$ denotes the 2D matrix covariate, again with all elements from standard normal, and $\Bbf \in \real{64 \times 64}$ is the matrix coefficient. $\Bbf$ takes the value of 0 or 1, and contains a series of shapes as shown in Figure~\ref{fig:shapes}, including ``square", ``T-shape", ``disk", ``triangle", and ``butterfly". Our goal is to recover those shapes in $\Bbf$ by inferring the association between $Y_{ij}$ and $\Xbf_{ij}$ after adjusting for $\Zbf_{ij}$.

\subsection{Signal Recovery} 
\label{sec:signal}

As the true signal in reality is hardly of an exact low rank structure, the tensor model (\ref{eqn:r-tensorreg}) and the associated tensor GEE (\ref{eqn:tgee}) essentially provide a low rank \emph{approximation} to the true signal. It is thus important to verify if such an approximation is adequate. We set $n=500$, $m=4$, and show both the tensor GEE estimates under various ranks and the corresponding BIC values (\ref{eqn:bic}) in Figure~\ref{fig:shapes}. We first assume that the correlation structure is correctly specified, and will study potential misspecification in the next section. In this setup, ``square" has the true rank equal to 1, ``T-shape" has the rank 2, and the remaining shapes have the highest possible rank 64. It is clearly seen from the figure that the tensor GEE offers a sound recovery of the true signal, even for the signals with high rank or natural shape, e.g., ``disk" and ``butterfly". In addition, the BIC seems to identify the correct or best approximate rank for all the signals. 

\begin{figure}
\begin{center}
\begin{tabular}{c}
\includegraphics[width=5.25in]{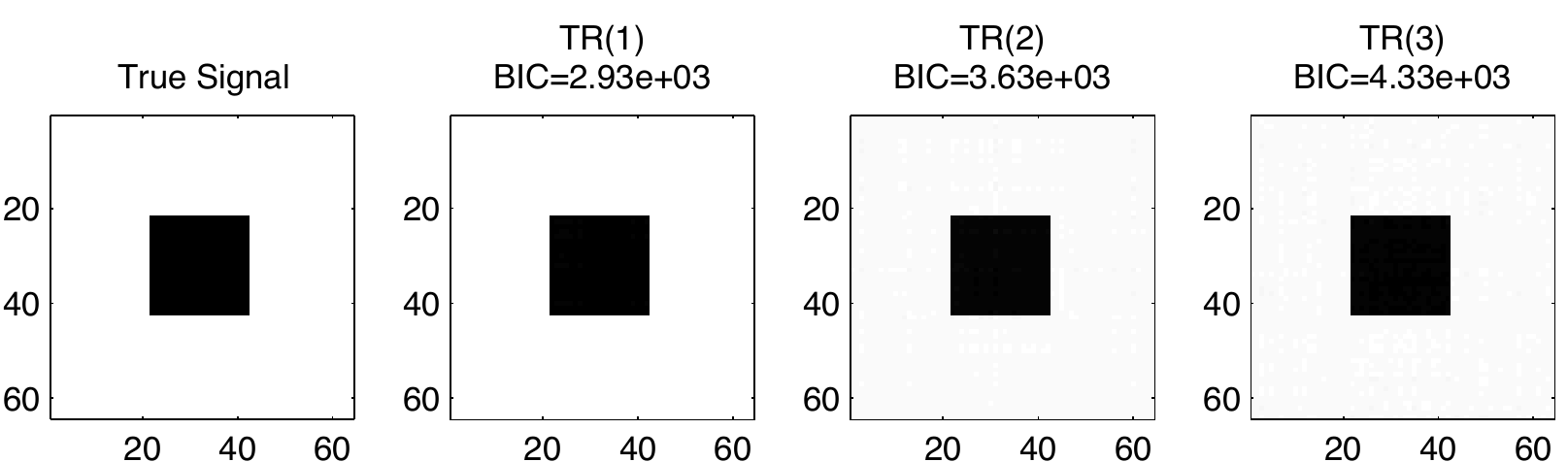} \\ 
\includegraphics[width=5.25in]{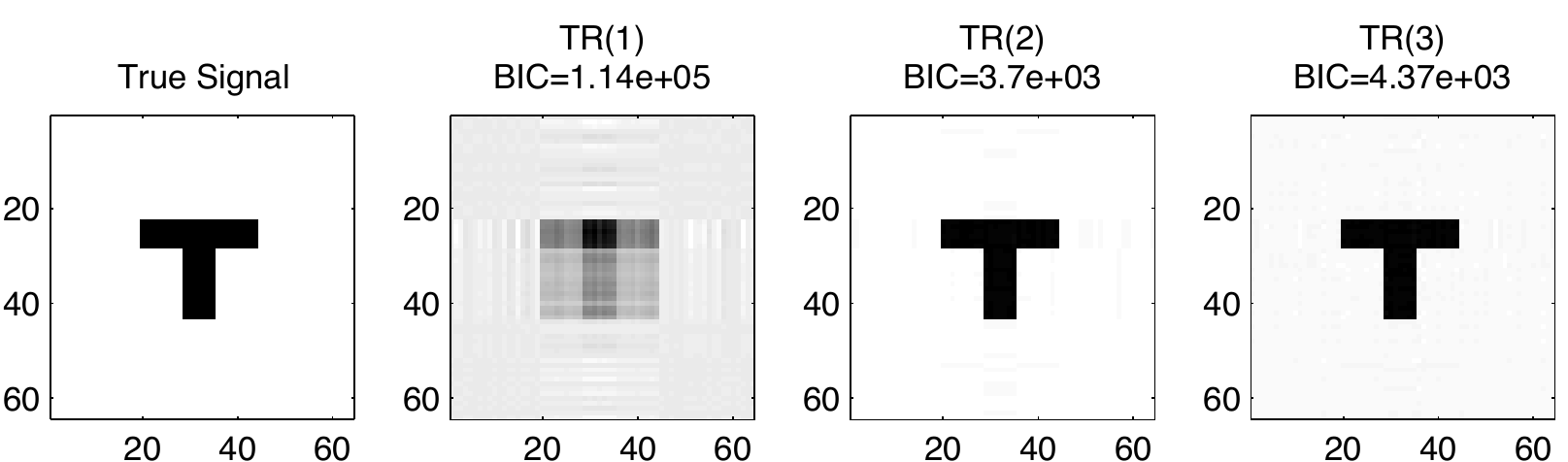} \\
\includegraphics[width=5.25in]{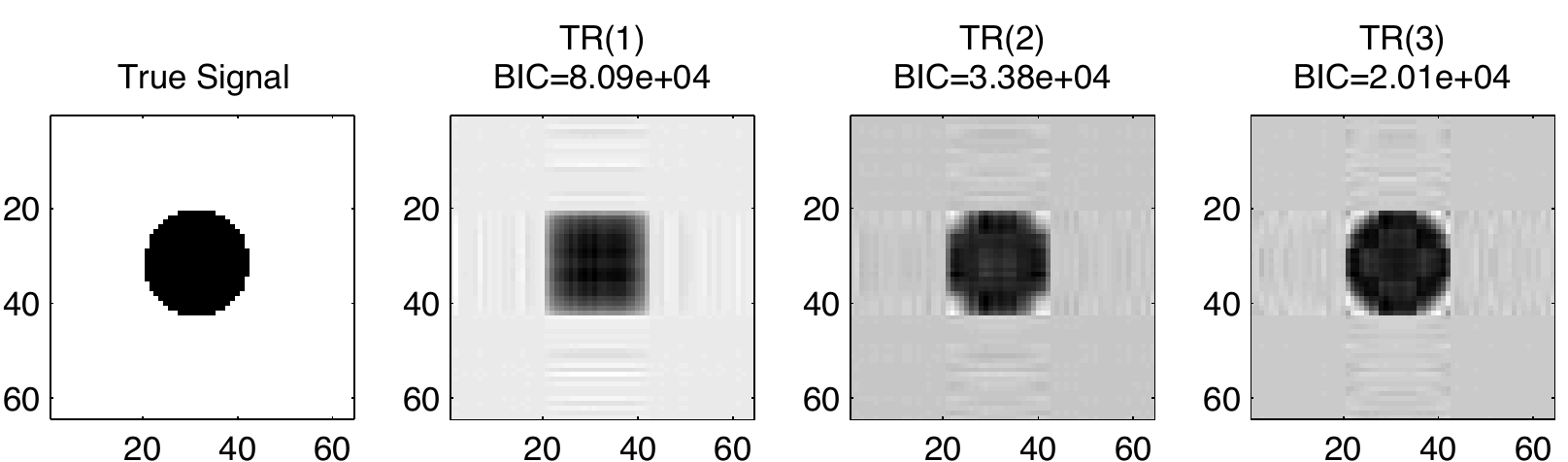} \\
\includegraphics[width=5.25in]{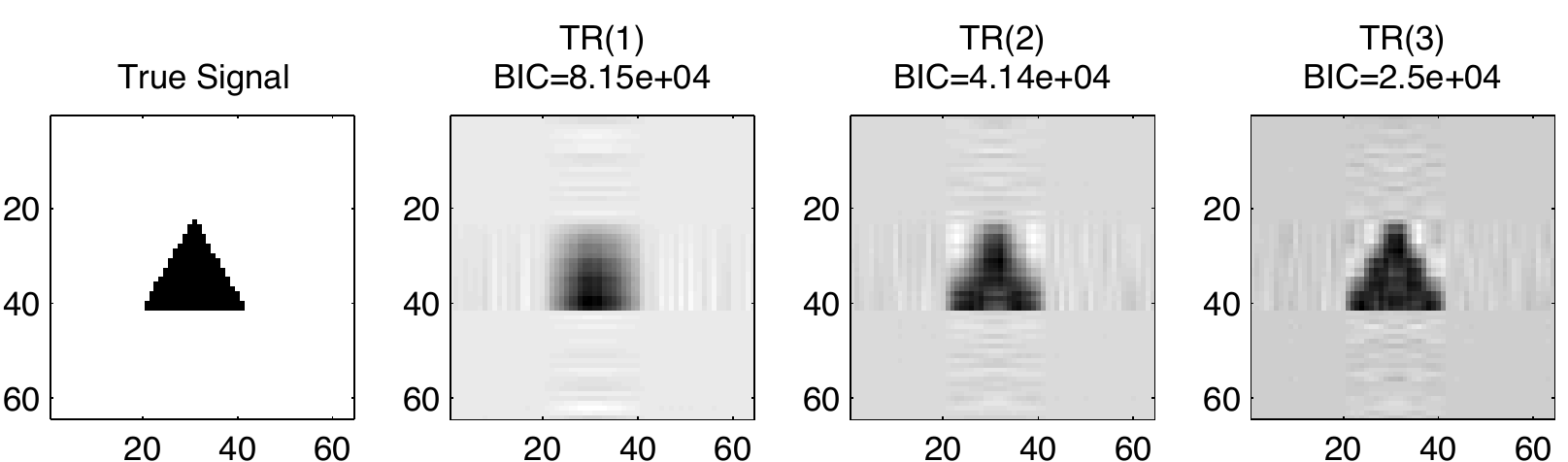} \\
\includegraphics[width=5.25in]{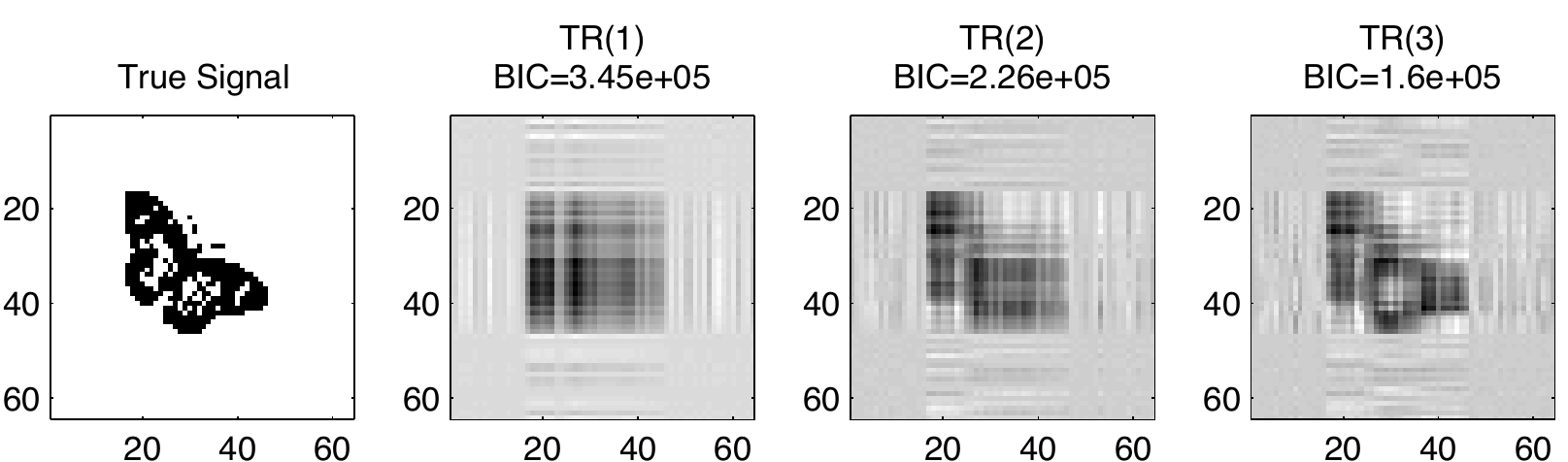}
\end{tabular}
\caption{True and recovered image signals by the tensor GEE with varying ranks. $n=500, m=4$. The correlation structure is correctly specified. TR$(R)$ means estimate from the rank-$R$ tensor model.}\label{fig:shapes}
\end{center}
\end{figure}

\subsection{Effect of Correlation Specification}

\begin{table}[t]
\caption{Bias, variance, and MSE of the tensor GEE estimates under various working correlation structures. Reported are the average out of 100 simulation replicates. The true intra-subject correlation is exchangeable with $\rho=0.8$.} \label{tab:correlation}
\begin{center}
\begin{tabular}{cccccc} 
\toprule
$n$ & $m$ & Working Correlation & Bias$^2$ & Variance & MSE \\
\midrule \midrule
50 & 10 & Exchangeable & 122.0& 383.6& {\bf 505.6(7.9)}\\
       &      & AR-1             &139.1 &530.0 &669.1(15.8)\\
      &      & Independence   &119.1 &393.9 & 513.0(11.0)\\
\midrule
100 & 10 & Exchangeable & 85.8& 128.9& {\bf 214.7(2.2)} \\
       &      & AR-1             &88.0 &159.1 &247.1(3.0)\\
      &      & Independence   &93.0 &141.2 &234.2(2.8)\\
\midrule
150 & 10 & Exchangeable &86.1 &51.3 & {\bf 137.2(0.6)}\\
       &      & AR-1           &85.6 &56.0 &141.6(0.6)\\
      &      & Independence   &84.9 & 62.3&147.2(0.9)\\ 
\bottomrule
\end{tabular}
\end{center}
\end{table}

We have shown that the tensor GEE estimator remains asymptotically consistent even when the working correlation structure is misspecified. However this describes only the \emph{large sample} behavior. In this section, we investigate potential effect of correlation misspecification when the sample size is \emph{small} or \emph{moderate}. 

We chose the ``butterfly" signal and fitted the tensor GEE model with three different working correlation structures: exchangeable, which is the correct specification in our setup, autoregressive of order one (AR-1), and independent. Table~\ref{tab:correlation} reports the averages and standard errors out of 100 replicates of the squared bias, the variance, and the mean squared error (MSE) of the tensor GEE estimate. We observe that the estimator based on the correct working correlation structure, i.e., the exchangeable structure, performs better than those based on misspecified correlation structures. When the sample size is moderate ($n=100$), all the estimators have comparable bias, while the difference in MSE mostly comes from the variance part of the estimator. This agrees with the theory that the choice of the working correlation structure affects the asymptotic variance of the estimator. When the sample size becomes relatively large ($n =150$), all the estimators perform similarly by the scaling term of $n^{-1/2}$ on the variance. When the sample size is small ($n=50$), all the estimators have relatively large bias, while the independence working structure yield similar results as the exchangeable structure. This suggests that, when the sample size is \emph{limited}, using a simple independence working structure is probably preferable compared to a more complex correlation structure.

\begin{figure}[t]
\begin{center}
\begin{tabular}{c}
\vspace{0.1in}
\includegraphics[width=5.in]{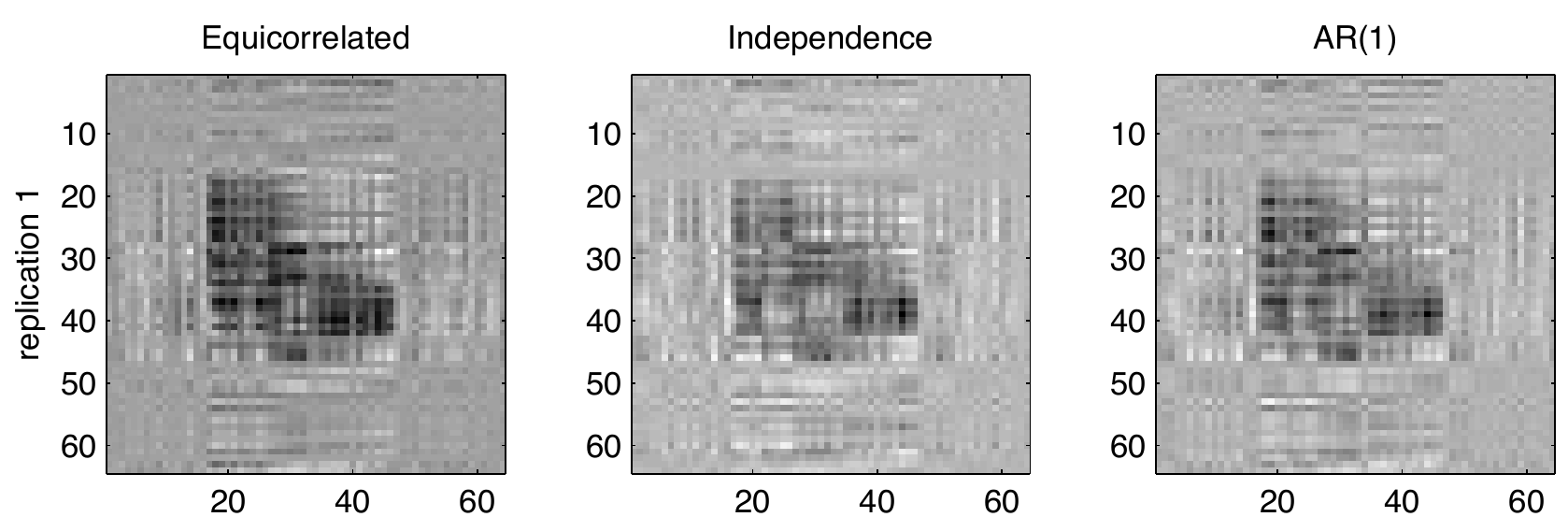} \\ 
\includegraphics[width=5.in]{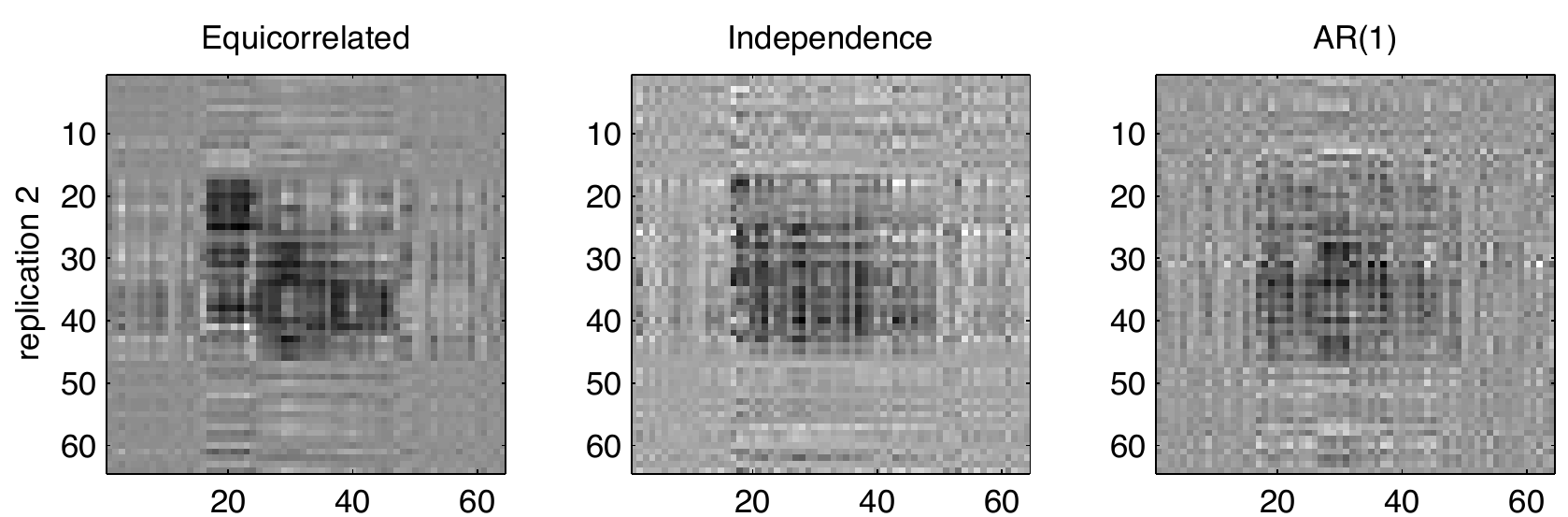} 
\end{tabular}
\caption{Snapshots of tensor GEE estimation with different working correlation structures. The true correlation is an equicorrelated structure. The comparison is \emph{row-wise}. The first row shows a replicate where the estimates are ``close" to the average behavior, and thus the visual quality of the estimates under different correlations structures are similar. The second row shows a replicate where the estimates are ``far away" from the average, then the estimate under the correct correlation structure (panel 1) is clearly superior than those under incorrect structures.}\label{fig:cor-snapshot}
\end{center}
\end{figure}

Nevertheless, we should bear in mind that the above observations are for the \emph{average} behavior of the estimate. Figure~\ref{fig:cor-snapshot} shows \emph{two snapshots} of the estimated signals under the three working correlations at $n=100$. The top panel is one replicate where the estimates are ``close" to the average in the sense that the bias, variance and MSE values for this single data realization are similar to those averages reported in Table~\ref{tab:correlation}. Consequently, the visual qualities of the three recovered signals are similar. The bottom panel, on the other hand, shows another replicate where the estimates are ``far away" from the average. Then the quality of the estimated signal under the correct working correlation structure is superior than the ones under the incorrect specifications. Such an observation suggests that, as long as the sample size of the study is moderate to large, a longitudinal model should be favored over the one that totally ignores potential intra-subject correlation.

\subsection{Regularized Estimation}
\label{sec:lasso}

We implemented the regularized tensor GEE with a lasso penalty, which extends the penalized GEE method of \citet{Wang2012} from vector to array covariate. It can identify relevant regions in images that are associated with the outcome, and this \emph{region selection} problem corresponds to \emph{variable selection} in classical vector covariate regressions. We studied the empirical performance by adopting the simulation setup described at the beginning of Section~\ref{sec:simulations}, but varying the sample size. The estimates of three shapes, ``T-shape", ``triangle", and ``butterfly", with and without regularizations, are shown in Figure~\ref{fig:regularization}. For the regularized tensor GEE, the penalty parameter $\lambda$ was selected based on the prediction accuracy on an independent validation set. It is clearly seen from the plot that, while increasing sample size improves estimation accuracy for both tensor GEE and regularized tensor GEE, regularization leads to a more accurate recovery, especially when the sample size is limited. As such we recommend the regularized tensor GEE for longitudinal imaging data analysis in practice. 

\begin{figure}
\begin{center}
\begin{tabular}{c}
\vspace{-0.3in}
\includegraphics[width=4.8in]{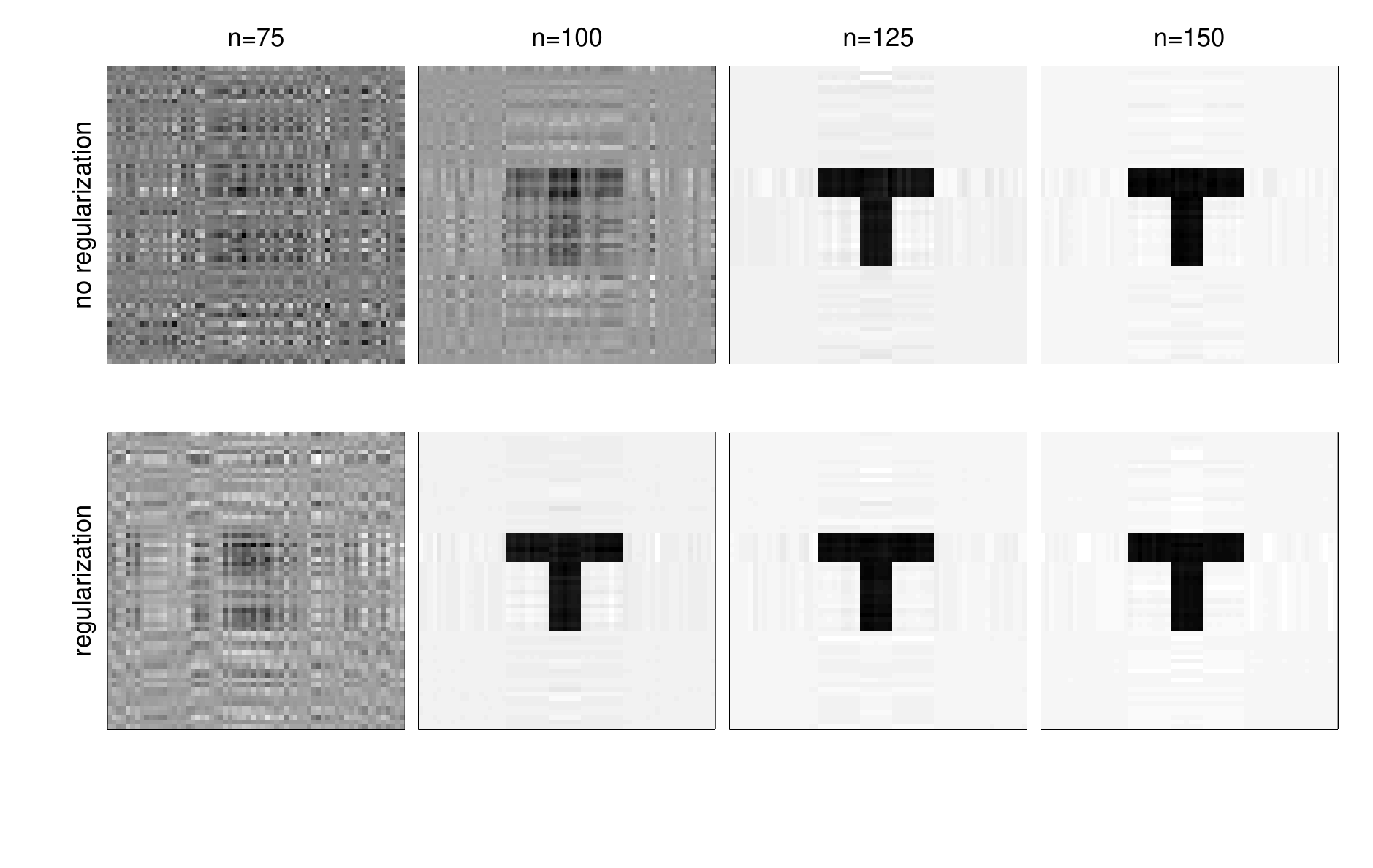} \\ \vspace{-0.3in}
\includegraphics[width=4.8in]{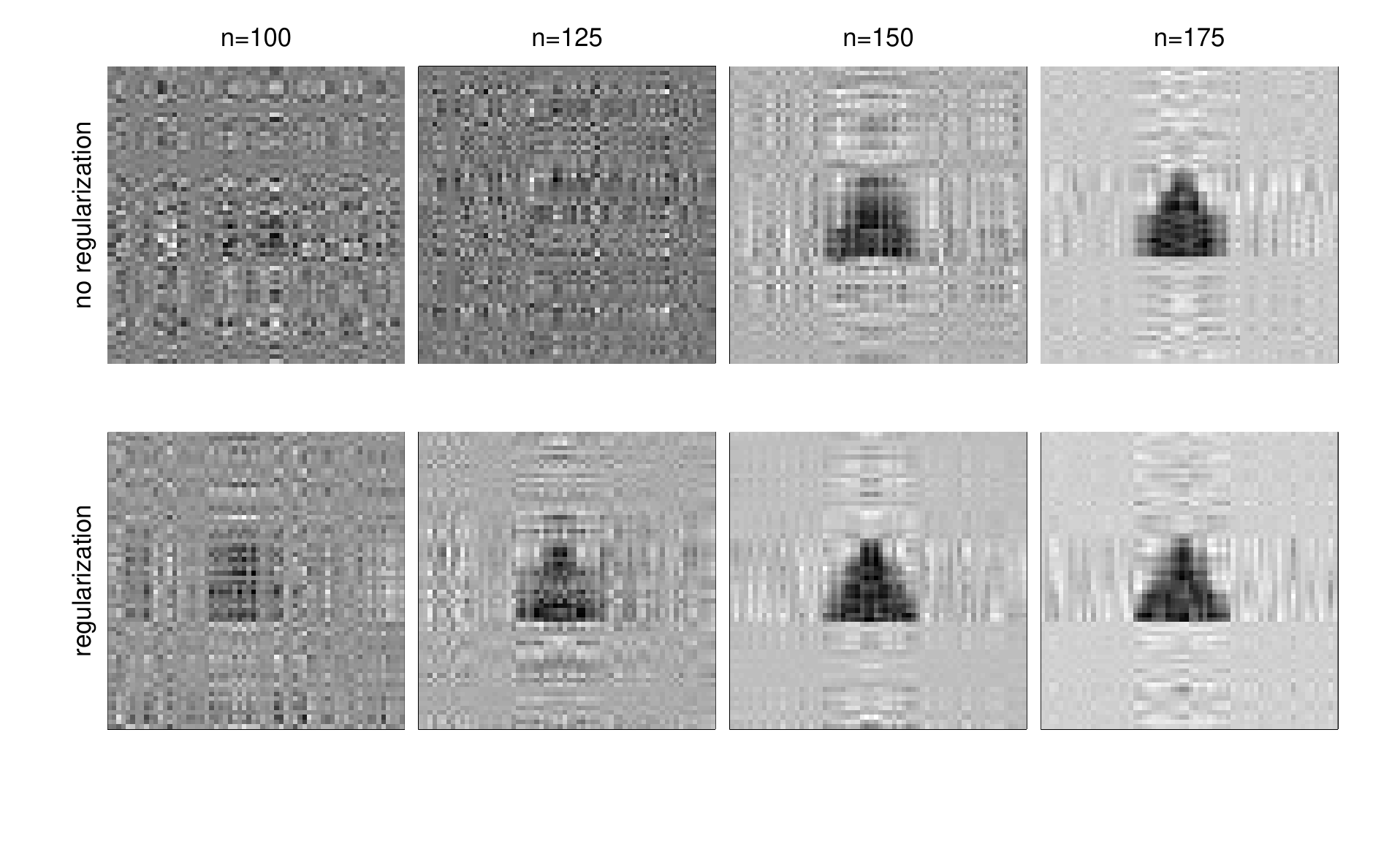} \\ \vspace{-0.3in}
\includegraphics[width=4.8in]{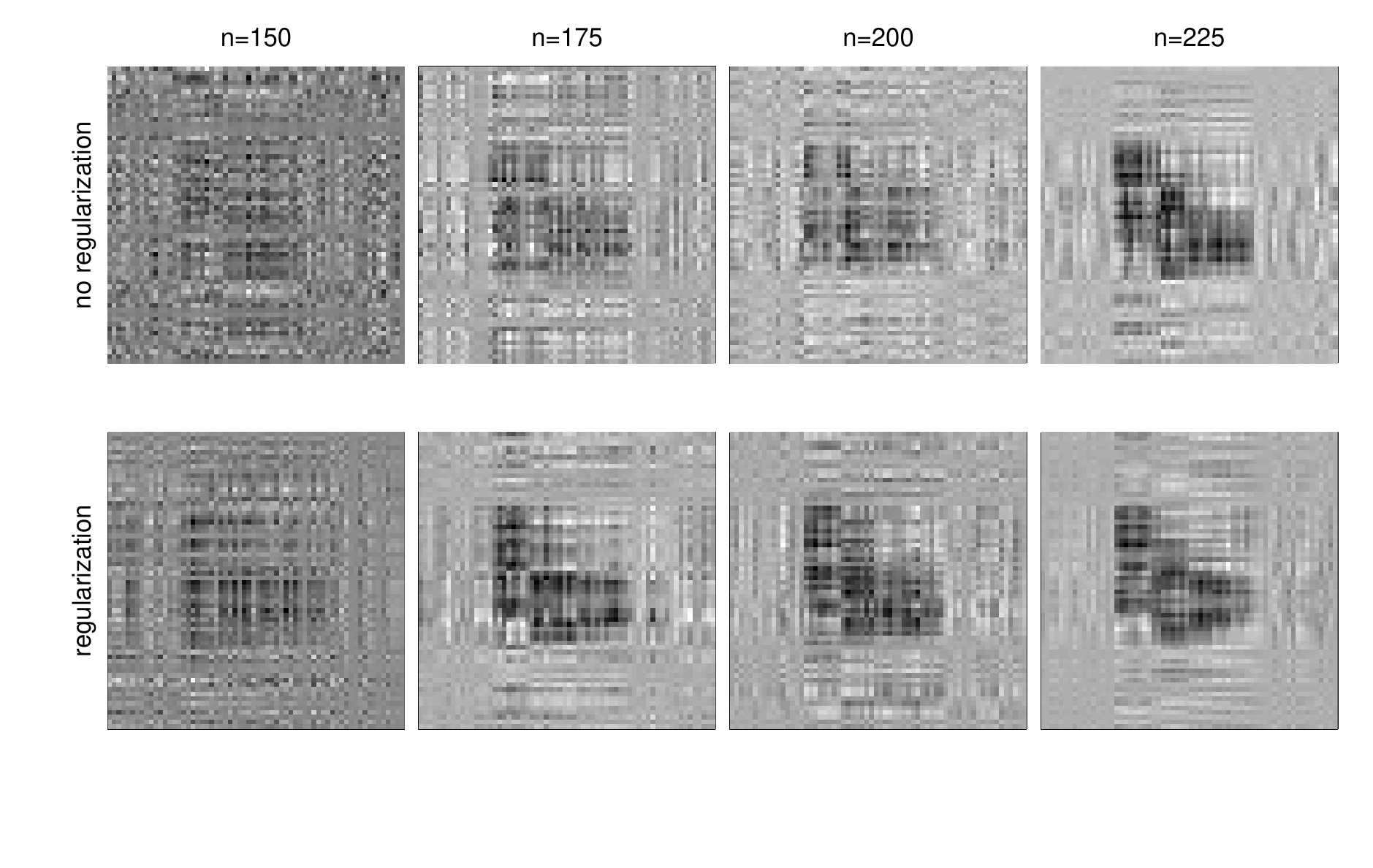} \\
\end{tabular}
\caption{Comparison of tensor GEE estimation with and without regularization under varying sample size. $m=4$. The matrix covariate is of size $64 \times 64$.}\label{fig:regularization}
\end{center}
\end{figure}

\section{Real Data Analysis}
\label{sec:realdata}

\subsection{Alzheimer's Disease}

Alzheimer's Disease (AD) is a progressive and irreversible neurodegenerative disorder and the leading form of dementia in elderly subjects. It is characterized by gradual impairment of cognitive and memory functions, and it has been projected to quadruple in its prevalence by the year 2050 \citep{Brookmeyer2007}. Amnestic mild cognitive impairment (MCI) is often a prodromal stage to Alzheimer's disease, and individuals with MCI may convert to AD at an annual rate as high as $15\%$ \citep{Petersen1999}. As such there is a pressing need for accurate and early diagnosis of AD and MCI, as well as monitoring their progression. The data we analyzed was obtained from the Alzheimer's Disease Neuroimaging Initiative (ADNI). It consists of $n=88$ MCI subjects with longitudinal MRI images of white matter at baseline, 6-month, 12-month, 18-month and 24-month ($m = 5$). Also recorded for each subject at multiple visits was the Mini Mental State Examination (MMSE) score. It measures the orientation to time and place, the immediate and delayed recall of three words, the attention and calculations, language, and visuoconstructional functions \citep{Folstein1975}, and is our response variable. A detailed description of acquiring MRI data from ADNI and the preprocessing protocol can be found in \citet{ZhangShen2012}. There are two scientific goals for this study. One is to predict the future clinical scores based on the data at previous time points, which is particularly useful for monitoring disease progression. The second is to identify brain subregions that are highly relevant to the disorder. We fitted tensor GEE to this data for both score prediction and region selection.

\subsection{Prediction and Disease Prognosis}

We downsized the original $256 \times 256 \times 256$ MRI images to $32 \times 32 \times 32$ via interpolation for computational simplicity. We first fitted tensor GEE using the data from baseline to 12-month, and used prediction of MMSE at 18-month to select the tuning parameter $\lambda$. Then we refitted the model using the data from baseline to 18-month under the selected $\lambda$, and evaluated the prediction accuracy of all subjects using the ``future" MMSE score at 24-month. The accuracy was evaluated by the rooted mean squared error (RMSE), $\{ n^{-1} \sum_{i=1}^n (Y_{im} - \hat Y_{im})^2 \}^{1/2}$, and the correlation, $\textrm{Corr}(Y_{im}, \hat Y_{im})$. This evaluation scheme is the same as that of \citet{ZhangShen2012}. Table~\ref{tab:predict} summarizes the results. It is seen that, for this data set, the best prediction was achieved under an AR(1) working correlation structure with $L_1$ regularization. The corresponding RMSE and correlation were 2.270 and 0.747, which are only slightly worse than the best reported RMSE 2.035 and correlation 0.786 in \citet{ZhangShen2012}. Note that \citet{ZhangShen2012} used multiple imaging modalities and additional clinical covariates, which are supposed to improve the prediction accuracy, while our study utilized only one imaging modality.

\begin{table}[t]
\caption{Prediction of future clinical MMSE scores using tensor GEE} \label{tab:predict}
\begin{center}
\begin{tabular}{lccccc} 
\toprule
 & \multicolumn{4}{c}{RMSE: $\{ \sum_{i=1}^nn^{-1} (Y_{im} - \hat Y_{im})^2 \}^{1/2}$ }\\
\cmidrule{2-5}
Working Correlation & Independence & Equicorrelated & AR(1) & Unstructured  \\ 
\midrule \midrule
regularization &2.460 &2.349 & {\bf 2.270} &2.570\\
 no regularization & 2.526&2.427 &2.429 &2.628\\
\midrule 
 & \multicolumn{4}{c}{Correlation: $\textrm{Corr}(Y_{im}, \hat Y_{im})$ }\\
\cmidrule{2-5}
Working Correlation & Independence & Equicorrelated & AR(1) & Unstructured  \\
\midrule \midrule
regularization &0.705 &0.733 & {\bf 0.747} &0.700\\
 no regularization & 0.701&0.716 &0.725 & 0.693\\
\bottomrule
\end{tabular}
\end{center}
\end{table}

\subsection{Region Selection}

We applied the lasso regularized tensor GEE to this data, and Figure~\ref{fig:sparse} shows the estimate (marked in red) overlaid on an image of an arbitrarily chosen subject, with three views, top, side and bottom, respectively. The identified anatomical regions mainly correspond to cerebral cortex, part of temporal lobe, parietal lobe, and frontal lobe \citep{Braak1991, Desikan2009, Yao2012}. With AD, patients experience significant widespread damage over the brain, causing shrinkage of brain volume \citep{Yao2012, Harasty1999} and thinning of cortical thickness \citep{Desikan2009, Yao2012}. The affected brain regions include those involved in controlling language (Broca's area) \citep{Harasty1999}, reasoning (superior and inferior frontal gyri) \citep{Harasty1999}, part of sensory area (primary auditory cortex, olfactory cortex, insula, and operculum) \citep{Braak1991, Lee2013}, somatosensory association area \citep{Yao2012, Tales2005, Mapstone2003}, memory loss (hippocampus) \citep{Heijer2010}, and motor function \citep{Buchman2011}. It is interesting to note that these regions are affected starting at different stages of AD, indicating the capability of the proposed method to locate brain atrophies as disease progresses. Specifically, hippocampus, which is highly correlated to memory loss, is commonly detected at the earliest stage of the disease. Regions related to language, communication, and motor functions are normally detected at the later stages of the disease. The fact that our findings are consistent with the results reported in previous studies demonstrates the efficacy of our proposed method in identifying correct biomarkers that are closely related to AD/MCI.

\begin{figure}[t]
\begin{center}
\begin{tabular}{ccc}
\includegraphics[width=1.5in]{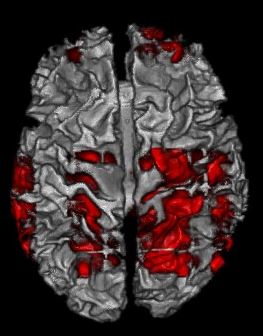}  & 
\includegraphics[width=1.75in]{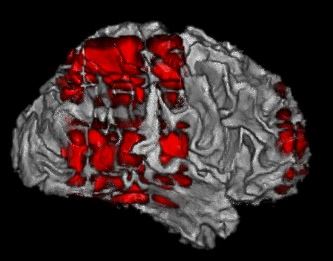} & 
\includegraphics[width=1.5in]{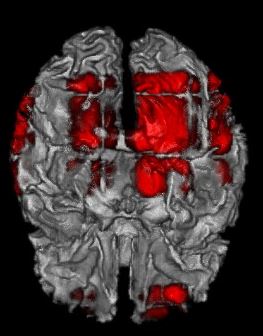}
\end{tabular}
\caption{The ADNI data: regularized estimate overlaid on a randomly selected subject.}\label{fig:sparse}
\end{center}
\end{figure}

\section{Discussions}
\label{sec:discussion}

We have proposed a tensor GEE approach for analyzing longitudinal imaging data. Our method combines the powerful GEE idea for handling longitudinal correlation and the low rank tensor decomposition to reduce the vast dimensionality of imaging data. The proposed algorithm scales well with imaging data size and is easy to implement using existing statistical softwares. Simulation studies and real data analysis show the advantage of our method for both signal recovering and prediction.

In the current paper, we have considered an image covariate together with a conventional vector covariate. Extending to joint multi-modality imaging analysis is conceptually easy: simply adding more array covariates into the systematic component \eqref{eqn:r-tensorreg}. However this brings up other issues such as joint selection of ranks for multiple array covariates, properly defining interactions between tensor covariates, and even higher volume of data. These important yet nontrivial questions deserve further investigation.

\baselineskip=15pt
%% used by Hua
%\bibliography{../../../bib-HZ}
%\bibliographystyle{asa}

% used by Lexin
\bibliography{ref-tensor}

\begin{thebibliography}{}

\bibitem[ADHD, 2014]{ADHD-url}
ADHD (2014).
\newblock The {ADHD}-200 sample.
\newblock \url{http://fcon_1000.projects.nitrc.org/indi/adhd200/}.
\newblock [Online; accessed 04-Mar-2014].

\bibitem[ADNI, 2014]{ADNI-url}
ADNI (2014).
\newblock Alzheimer's disease neuroimaging initiative.
\newblock \url{http://adni.loni.usc.edu}.
\newblock [Online; accessed 04-Mar-2014].

\bibitem[Aston and Kirch, 2012]{AstonKirch2012}
Aston, J.~A. and Kirch, C. (2012).
\newblock Estimation of the distribution of change-points with application to
  fmri data.
\newblock {\em Annals of Applied Statistics}, 6:1906--1948.

\bibitem[Balan and Schiopu-Kratina, 2005]{Balan2005}
Balan, R.~M. and Schiopu-Kratina, I. (2005).
\newblock Asymptotic results with generalized estimating equations for
  longitudinal data.
\newblock {\em The Annals of Statistics}, 33(2):522--541.

\bibitem[Braak and Braak, 1991]{Braak1991}
Braak, H. and Braak, E. (1991).
\newblock Neuropathological stageing of {Alzheimer-related} changes.
\newblock {\em Acta Neuropathologica}, 82(4):239--259.

\bibitem[Brookmeyer et~al., 2007]{Brookmeyer2007}
Brookmeyer, R., Johnson, E., Ziegler-Graham, K., and Arrighi, H.~M. (2007).
\newblock Forecasting the global burden of alzheimer’s disease.
\newblock {\em Alzheimer's \& Dementia}, 3(3):186 -- 191.

\bibitem[Buchman and Bennett, 2011]{Buchman2011}
Buchman, A. and Bennett, D. (2011).
\newblock Loss of motor function in preclinical alzheimer's disease.
\newblock {\em Expert Review Neurotherapeutics}, 11(5):665--676.

\bibitem[Caffo et~al., 2010]{Caffo2010}
Caffo, B., Crainiceanu, C., Verduzco, G., Joel, S., S.H., M., Bassett, S., and
  Pekar, J. (2010).
\newblock Two-stage decompositions for the analysis of functional connectivity
  for {fMRI} with application to {A}lzheimer's disease risk.
\newblock {\em NeuroImage}, 51(3):1140--1149.

\bibitem[Davatzikos et~al., 2009]{Davatzikos09Longitudinal}
Davatzikos, C., Xu, F., An, Y., Fan, Y., and Resnick, S.~M. (2009).
\newblock Longitudinal progression of alzheimer's-like patterns of atrophy in
  normal older adults: the spare-ad index.
\newblock {\em Brain}, 132(8):2026--2035.

\bibitem[den Heijer et~al., 2010]{Heijer2010}
den Heijer, T., van~der Lijn, F., Koudstaal, P.~J., Hofman, A., van~der Lugt,
  A., Krestin, G.~P., Niessen, W.~J., and Breteler, M. M.~B. (2010).
\newblock A 10-year follow-up of hippocampal volume on magnetic resonance
  imaging in early dementia and cognitive decline.
\newblock {\em Brain}, 133(4):1163--1172.

\bibitem[Desikan et~al., 2009]{Desikan2009}
Desikan, R., Cabral, H., Hess, C., Dillon, W., Salat, D., Buckner, R., Fischl,
  B., and Initiative, A. D.~N. (2009).
\newblock Automated {MRI} measures identify individuals with mild cognitive
  impairment and {Alzheimer's} disease.
\newblock {\em Brain}, 132:2048--2057.

\bibitem[Fan and Li, 2001]{FanLi01SCAD}
Fan, J. and Li, R. (2001).
\newblock Variable selection via nonconcave penalized likelihood and its oracle
  properties.
\newblock {\em J. Amer. Statist. Assoc.}, 96(456):1348--1360.

\bibitem[Fan and Li, 2004]{FanLi2004}
Fan, J. and Li, R. (2004).
\newblock New estimation and model selection procedures for semiparametric
  modeling in longitudinal data analysis.
\newblock {\em Journal of the American Statistical Association},
  99(467):710--723.

\bibitem[Folstein et~al., 1975]{Folstein1975}
Folstein, M.~F., Folstein, S.~E., and McHugh, P.~R. (1975).
\newblock Mini-mental state: A practical method for grading the cognitive state
  of patients for the clinician.
\newblock {\em Journal of Psychiatric Research}, 12(3):189 -- 198.

\bibitem[Frank and Friedman, 1993]{FrankFriedman93Bridge}
Frank, I.~E. and Friedman, J.~H. (1993).
\newblock A statistical view of some chemometrics regression tools.
\newblock {\em Technometrics}, 35(2):109--135.

\bibitem[Friston, 2009]{Friston2009}
Friston, K.~J. (2009).
\newblock Modalities, modes, and models in functional neuroimaging.
\newblock {\em Science}, 326:399--403.

\bibitem[Harasty et~al., 1999]{Harasty1999}
Harasty, J.~A., Halliday, G.~M., Kril, J.~J., and Code, C. (1999).
\newblock Specific temporoparietal gyral atrophy reflects the pattern of
  language dissolution in alzheimer's disease.
\newblock {\em Brain}, 122(4):675--686.

\bibitem[Hinrichs et~al., 2009]{Hinrichs2009}
Hinrichs, C., Singh, V., Mukherjee, L., Xu, G., Chung, M.~K., Johnson, S.~C.,
  and ADNI (2009).
\newblock Spatially augmented lpboosting for ad classification with evaluations
  on the adni dataset.
\newblock {\em NeuroImage}, 48:138--149.

\bibitem[Hinrichs et~al., 2011]{Hinrichs11Predict}
Hinrichs, C., Singh, V., Xu, G., and Johnson, S.~C. (2011).
\newblock Predictive markers for \{AD\} in a multi-modality framework: An
  analysis of \{MCI\} progression in the \{ADNI\} population.
\newblock {\em NeuroImage}, 55(2):574 -- 589.

\bibitem[Kang et~al., 2012]{KangOmbao2012}
Kang, H., Ombao, H., Linkletter, C., Long, N., and Badre, D. (2012).
\newblock Spatio-spectral mixed-effects model for functional magnetic resonance
  imaging data.
\newblock {\em Journal of the American Statistical Association},
  107(498):568--577.

\bibitem[Kolda and Bader, 2009]{KoldaBader09Tensor}
Kolda, T.~G. and Bader, B.~W. (2009).
\newblock Tensor decompositions and applications.
\newblock {\em SIAM Rev.}, 51(3):455--500.

\bibitem[Lazar, 2008]{Lazar2008}
Lazar, N.~A. (2008).
\newblock {\em The Statistical Analysis of Functional MRI Data}.
\newblock Springer, New York.

\bibitem[Lee et~al., 2013]{Lee2013}
Lee, T.~M., Sun, D., Leung, M.-K., Chu, L.-W., and Keysers, C. (2013).
\newblock Neural activities during affective processing in people with
  alzheimer's disease.
\newblock {\em Neurobiology of Aging}, 34(3):706 -- 715.

\bibitem[Li, 1997]{LiB1997}
Li, B. (1997).
\newblock On the consistency of generalized estimating equations.
\newblock In {\em Selected {P}roceedings of the {S}ymposium on {E}stimating
  {F}unctions ({A}thens, {GA}, 1996)}, volume~32 of {\em IMS Lecture Notes
  Monogr. Ser.}, pages 115--136. Inst. Math. Statist., Hayward, CA.

\bibitem[Li et~al., 2013]{LiShen2013}
Li, Y., Gilmore, J.~H., Shen, D., Styner, M., Lin, W., and Zhu, H. (2013).
\newblock Multiscale adaptive generalized estimating equations for longitudinal
  neuroimaging data.
\newblock {\em NeuroImage}, 72(0):91 -- 105.

\bibitem[Liang and Zeger, 1986]{LiangZeger1986}
Liang, K.~Y. and Zeger, S.~L. (1986).
\newblock Longitudinal data analysis using generalized linear models.
\newblock {\em Biometrika}, 73(1):13--22.

\bibitem[Mapstone et~al., 2003]{Mapstone2003}
Mapstone, M., Steffenella, T., and Duffy, C. (2003).
\newblock A visuospatial variant of mild cognitive impairment: getting lost
  between aging and ad.
\newblock {\em Neurology}, 60:802--808.

\bibitem[McEvoy et~al., 2011]{McEvoy11Longitudinal}
McEvoy, L.~K., Holland, D., Hagler, D.~J., Fennema-Notestine, C., Brewer,
  J.~B., and Dale, A.~M. (2011).
\newblock Mild cognitive impairment: Baseline and longitudinal structural mr
  imaging measures improve predictive prognosis.
\newblock {\em Radiology}, 259(3):834--843.
\newblock PMID: 21471273.

\bibitem[Misra et~al., 2009]{Misra09LongtitudinalPrediction}
Misra, C., Fan, Y., and Davatzikos, C. (2009).
\newblock Baseline and longitudinal patterns of brain atrophy in \{MCI\}
  patients, and their use in prediction of short-term conversion to ad: Results
  from \{ADNI\}.
\newblock {\em NeuroImage}, 44(4):1415 -- 1422.

\bibitem[Ni et~al., 2010]{Ni2010}
Ni, X., Zhang, D., and Zhang, H.~H. (2010).
\newblock Variable selection for semiparametric mixed models in longitudinal
  studies.
\newblock {\em Biometrics}, 66(1):79--88.

\bibitem[Ortega and Rheinboldt, 2000]{ortega2000iterative}
Ortega, J.~M. and Rheinboldt, W.~C. (2000).
\newblock {\em Iterative solution of nonlinear equations in several variables},
  volume~30.
\newblock Siam.

\bibitem[Pan, 2001]{Pan2001}
Pan, W. (2001).
\newblock Akaike's information criterion in generalized estimating equations.
\newblock {\em Biometrics}, 57(1):120--125.

\bibitem[Petersen et~al., 1999]{Petersen1999}
Petersen, R., Smith, G., Waring, S., Ivnik, R., Tangalos, E., and Kokmen, E.
  (1999).
\newblock Mild cognitive impairment: clinical characterization and outcome.
\newblock {\em Archives of Neurology}, 56:303--308.

\bibitem[Prentice and Zhao, 1991]{Prentice1991}
Prentice, R.~L. and Zhao, L.~P. (1991).
\newblock Estimating equations for parameters in means and covariances of
  multivariate discrete and continuous responses.
\newblock {\em Biometrics}, 47(3):825--839.

\bibitem[Qu et~al., 2000]{Qu2000}
Qu, A., Lindsay, B.~G., and Li, B. (2000).
\newblock Improving generalised estimating equations using quadratic inference
  functions.
\newblock {\em Biometrika}, 87(4):823--836.

\bibitem[Rao and Mitra, 1971]{RaoMitra71GenInv}
Rao, C.~R. and Mitra, S.~K. (1971).
\newblock {\em Generalized Inverse of Matrices and its Applications}.
\newblock John Wiley\thinspace \&\thinspace Sons, Inc., New York-London-Sydney.

\bibitem[Reiss and Ogden, 2010]{ReissOgden2010}
Reiss, P. and Ogden, R. (2010).
\newblock Functional generalized linear models with images as predictors.
\newblock {\em Biometrics}, 66:61--69.

\bibitem[Skup et~al., 2012]{Skup2012}
Skup, M., Zhu, H., and Zhang, H. (2012).
\newblock Multiscale adaptive marginal analysis of longitudinal neuroimaging
  data with time-varying covariates.
\newblock {\em Biometrics}, 68(4):1083--1092.

\bibitem[Song et~al., 2009]{SongQu2009}
Song, P. X.-K., Jiang, Z., Park, E., and Qu, A. (2009).
\newblock Quadratic inference functions in marginal models for longitudinal
  data.
\newblock {\em Statistics in Medicine}, 28(29):3683--3696.

\bibitem[Tales et~al., 2005]{Tales2005}
Tales, A., Haworth, J., Nelson, S., J.~Snowden, R., and Wilcock, G. (2005).
\newblock Abnormal visual search in mild cognitive impairment and alzheimer's
  disease.
\newblock {\em Neurocase}, 11(1):80--84.

\bibitem[Tibshirani, 1996]{Tibshirani96Lasso}
Tibshirani, R. (1996).
\newblock Regression shrinkage and selection via the lasso.
\newblock {\em J. Roy. Statist. Soc. Ser. B}, 58(1):267--288.

\bibitem[Wang, 2011]{Wang2011}
Wang, L. (2011).
\newblock G{EE} analysis of clustered binary data with diverging number of
  covariates.
\newblock {\em The Annals of Statistics}, 39(1):389--417.

\bibitem[Wang et~al., 2012]{Wang2012}
Wang, L., Zhou, J., and Qu, A. (2012).
\newblock Penalized generalized estimating equations for high-dimensional
  longitudinal data analysis.
\newblock {\em Biometrics}, 68(2):353--360.

\bibitem[Wang et~al., 2014]{Wang2014}
Wang, X., Nan, B., Zhu, J., and Koeppe, R. (2014).
\newblock Regularized {3D} functional regression for brain image data via haar
  wavelets.
\newblock {\em The Annals of Applied Statistics}, page in press.

\bibitem[Xie and Yang, 2003]{Xie2003}
Xie, M. and Yang, Y. (2003).
\newblock Asymptotics for generalized estimating equations with large cluster
  sizes.
\newblock {\em The Annals of Statistics}, 31(1):310--347.

\bibitem[Xue et~al., 2010]{Xue2010}
Xue, L., Qu, A., and Zhou, J. (2010).
\newblock Consistent model selection for marginal generalized additive model
  for correlated data.
\newblock {\em Journal of the American Statistical Association},
  105(492):1518--1530.
\newblock Supplementary materials available online.

\bibitem[Yao et~al., 2012]{Yao2012}
Yao, Z., Hu, B., Liang, C., Zhao, L., Jackson, M., and the Alzheimer's Disease
  Neuroimaging~Initiative (2012).
\newblock A longitudinal study of atrophy in amnestic mild cognitive impairment
  and normal aging revealed by cortical thickness.
\newblock {\em PLoS One}, 7(11):e48973.

\bibitem[Zhang, 2010]{zhang2010nearly}
Zhang, C.-H. (2010).
\newblock Nearly unbiased variable selection under minimax concave penalty.
\newblock {\em The Annals of Statistics}, 38(2):894--942.

\bibitem[Zhang and Huang, 2008]{zhang2008sparsity}
Zhang, C.-H. and Huang, J. (2008).
\newblock The sparsity and bias of the lasso selection in high-dimensional
  linear regression.
\newblock {\em The Annals of Statistics}, pages 1567--1594.

\bibitem[Zhang et~al., 2012]{ZhangShen2012}
Zhang, D., Shen, D., and {Alzheimer's Disease Neuroimaging Initiative} (2012).
\newblock Predicting future clinical changes of mci patients using longitudinal
  and multimodal biomarkers.
\newblock {\em PLoS One}, 7(3):e33182.

\bibitem[Zhang et~al., 2011]{ZhangShen2011}
Zhang, D., Wang, Y., Zhou, L., Yuan, H., Shen, D., and the Alzheimers Disease
  Neuroimaging~Initiative (2011).
\newblock Multimodal classification of {Alzheimer's} disease and mild cognitive
  impairment.
\newblock {\em NeuroImage}, 55(3):856 -- 867.

\bibitem[Zhou et~al., 2013]{ZhouLiZhu2013}
Zhou, H., Li, L., and Zhu, H. (2013).
\newblock Tensor regression with applications in neuroimaging data analysis.
\newblock {\em Journal of the American Statistical Association},
  108(502):540--552.

\bibitem[Zou and Hastie, 2005]{ZouHastie05Enet}
Zou, H. and Hastie, T. (2005).
\newblock Regularization and variable selection via the elastic net.
\newblock {\em J. R. Stat. Soc. Ser. B Stat. Methodol.}, 67(2):301--320.

\end{thebibliography}
\bibliographystyle{apalike}

\newpage
\baselineskip=21pt
\section*{Appendix: Technical Proofs}

\noindent
\textbf{Outline of the proofs} \\
\noindent
We prove the results for the diverging case (Theorem 3 and Theorem 4) in the appendix. One can prove the results for the fixed case (Theorem 1 and Theorem 2) by using the same techniques below and replacing $p_n$ with a fixed positive constant. 

The proof of Lemma 1 is similar to the one of Theorem 3 by dropping the terms involving the working correlation matrix and thus is omitted here. 

To facilitate the proof, we introduce the following notations. Denote $\widehat{\boldsymbol\beta}_n=\boldsymbol\beta_{\hat{\Bbf}}$ and $\boldsymbol\beta_0=\betabf_{\Bbf_0}$. Recall that the CP decomposition ensures that $\Bbf$ is uniquely determined by $\boldsymbol\beta_n \in \real{R \sum_{d=1}^D p_d}$. Denote $\Jbf(\boldsymbol\beta)=[\Jbf_1 \cdots \Jbf_D]$, and note that under tensor structure $\partial \theta_{ij}/\partial \boldsymbol\beta=\Jbf(\boldsymbol\beta)\text{vec}(\Xbf_{ij})$. Recall the generalized estimating equations without vector covariates can be written as 
\begin{align*}
\sbf_n(\boldsymbol\beta_n)&=\sum_{i=1}^n\Jbf\trans(\boldsymbol\beta_n)\text{vec}(\Xbf_i)\Abf^{1/2}_i(\boldsymbol\beta_n)\widehat{\Rbf}^{-1}\Abf^{-1/2}_i(\boldsymbol\beta_n)(\Ybf_i-\boldsymbol\mu_i(\boldsymbol\beta_n)).
\end{align*}

The main technique to prove Theorem 3 is the sufficient condition for existence and consistency of a root of equations proposed in \citet{ortega2000iterative}. To check this condition, the following Lemma \ref{Sbar} - \ref{Dn1} are proposed. Lemma \ref{Sbar} provides a useful approximation to the generalized estimating equations $\sbf_n(\boldsymbol\beta_0)$ based on the Condition (A4*) of the working correlation matrix. This facilitates the later evaluations of the moments of the generalized estimating equations by treating the intra-subject correlation as known. Lemma \ref{Dbar} further establishes the approximation to the negative gradients of the generalized estimating equations. Lemma \ref{Dn1} refines this approximation to the negative gradients at one more step, providing the foundations for the Talyor expansion of generalized estimating equations at the true value. 

Based on Theorem 3, the proof of Theorem 4 is straightforward by evaluating the covariance matrix of the generalized estimating equations and applying the Lindeberg-Feller central limit theorem. 
 
\begin{lemma}\label{Sbar}
Under Conditions (A1*)-(A9*), $p_n=o(n^{1/2})$, then $||\tilde{\sbf}_n(\boldsymbol\beta_0)-\sbf_n(\boldsymbol\beta_0)||=O_p(p_n)$, where $\tilde{\sbf}_n(\boldsymbol\beta_0)$ is $\sbf_n(\boldsymbol\beta_0)$ with $\widehat{\Rbf}$ replaced by $\tilde{\Rbf}$.
\end{lemma}
\begin{proof}[Proof of Lemma \ref{Sbar}]
Consider 
\begin{align*}
\tilde{\sbf}_n(\boldsymbol\beta_n)&=\sum_{i=1}^n\Jbf\trans(\boldsymbol\beta_n)\text{vec}(\Xbf_i)\Abf^{1/2}_i(\boldsymbol\beta_n)\tilde{\Rbf}^{-1}\Abf^{-1/2}_i(\boldsymbol\beta_n)(\Ybf_i-\boldsymbol\mu_i(\boldsymbol\beta_n)).
\end{align*}
Denote $\{r_{i,j}\}_{1\le i,  j \le m}$ the $(i,j)$-th element of $\widehat{\Rbf}^{-1}-\tilde{\Rbf}^{-1}$. By Condition (A4*), $r_{i,j}=O_p(\sqrt{p_n/n})$. Note that
\begin{align*}
&\sbf_n(\boldsymbol\beta_0)-\tilde{\sbf}_n(\boldsymbol\beta_0)\\
=&\sum_{i=1}^n\sum_{j=1}^m\sum_{k=1}^mr_{j,m}\sigma_{ij}(\boldsymbol\beta_0)\epsilon_{ik}(\boldsymbol\beta_0)\Jbf\trans(\boldsymbol\beta_0)\text{vec}(\Xbf_{ij})\\
=&\sum_{j=1}^m\sum_{k=1}^mr_{j,m}\Big[\sum_{i=1}^n\sigma_{ij}(\boldsymbol\beta_0)\epsilon_{ik}(\boldsymbol\beta_0)\Jbf\trans(\boldsymbol\beta_0)\text{vec}(\Xbf_{ij})\Big],
\end{align*}
where $\epsilon_{ik}(\boldsymbol\beta_0)=\sigma_{ik}^{-1}(\boldsymbol\beta_0)(Y_{ik}-\mu_{ik}(\boldsymbol\beta_0))$. By Condition (A6*), $\mathbb{E}[\epsilon_{ik}(\boldsymbol\beta_0)]=O_p(1)$. Note that for any $1 \le j, k \le m$, 
\begin{align*}
&\mathbb{E}\Big[||\sum_{i=1}^n\sigma_{ij}(\boldsymbol\beta_0)\epsilon_{ik}(\boldsymbol\beta_0)\Jbf\trans(\boldsymbol\beta_n)\text{vec}(\Xbf_{ij})||^2\Big]\\
=&\sum_{i=1}^n\sigma^2_{ij}(\boldsymbol\beta_0)\mathbb{E}[\epsilon_{ik}^2(\boldsymbol\beta_0)]\text{vec}\trans(\Xbf_{ij})\Jbf(\boldsymbol\beta_0)\Jbf\trans(\boldsymbol\beta_0)\text{vec}(\Xbf_{ij})\\
=&\sum_{i=1}^n\sigma^2_{ij}(\boldsymbol\beta_0)\mathbb{E}[\epsilon_{ik}^2(\boldsymbol\beta_0)]\text{Tr}(\Jbf\trans(\boldsymbol\beta_0)\text{vec}(\Xbf_{ij})\text{vec}\trans(\Xbf_{ij})\Jbf(\boldsymbol\beta_0))\\
\le &Cnp_n,
\end{align*}
for some constant $C>0$ by Condition (A1*), (A2*) and (A7*). Since $r_{i,j}=O_p(\sqrt{p_n/n})$, the proof is complete.
\end{proof}

Consider $\Dbf_n(\boldsymbol\beta_n)=-\partial \sbf_n(\boldsymbol\beta_n)/\partial \boldsymbol\beta_n$, $\tilde{\Dbf}_n(\boldsymbol\beta_n)=-\partial \tilde{\sbf}_n(\boldsymbol\beta_n)/\partial \boldsymbol\beta_n$. Lemma \ref{Dbar} establishes the approximation to the negative gradients of the estimating equations. 

\begin{lemma}\label{Dbar}
Under Conditions (A1*)-(A9*), for any $\triangle>0$, 
\begin{align*}
\sup_{||\boldsymbol\beta_n-\boldsymbol\beta_0||\le \triangle \sqrt{p_n/n}}|\lambda_{\max}[\tilde{\Dbf}_n(\boldsymbol\beta_n)-\Dbf_n(\boldsymbol\beta_n)]|=O_p(\sqrt{p_nn}),\\
\sup_{||\boldsymbol\beta_n-\boldsymbol\beta_0||\le \triangle \sqrt{p_n/n}}|\lambda_{\min}[\tilde{\Dbf}_n(\boldsymbol\beta_n)-\Dbf_n(\boldsymbol\beta_n)]|=O_p(\sqrt{p_nn}).
\end{align*}
\end{lemma}
\begin{proof}[Proof of Lemma \ref{Dbar}]
Similar to Lemma C.1. of \citet{Wang2011}, it can be shown by direct calculation that
\begin{equation*}
\tilde{\Dbf}_n(\boldsymbol\beta_n)=\tilde{\Dbf}_{n1}(\boldsymbol\beta_n)+\tilde{\Dbf}_{n2}(\boldsymbol\beta_n)+\tilde{\Dbf}_{n3}(\boldsymbol\beta_n)+\tilde{\Dbf}_{n4}(\boldsymbol\beta_n),
\end{equation*}
where
\begin{align*}
\tilde{\Dbf}_{n1}(\boldsymbol\beta_n)&=\sum_{i=1}^n\Jbf\trans(\boldsymbol\beta_n)\text{vec}(\Xbf_i)\Abf^{1/2}_i(\boldsymbol\beta_n)\tilde{\Rbf}^{-1}\Abf^{1/2}_i(\boldsymbol\beta_n)\text{vec}\trans(\Xbf_i)\Jbf(\boldsymbol\beta_n),\\
\tilde{\Dbf}_{n2}(\boldsymbol\beta_n)&=\frac{1}{2}\sum_{i=1}^n\Jbf\trans(\boldsymbol\beta_n)\text{vec}(\Xbf_i)\Abf^{1/2}_i(\boldsymbol\beta_n)\tilde{\Rbf}^{-1}\Abf^{-3/2}_i(\boldsymbol\beta_n){\bf C}_i(\boldsymbol\beta_n)\Fbf_i(\boldsymbol\beta_n)\text{vec}\trans(\Xbf_i)\Jbf(\boldsymbol\beta_n),\\
\tilde{\Dbf}_{n3}(\boldsymbol\beta_n)&=-\frac{1}{2}\sum_{i=1}^n\Jbf\trans(\boldsymbol\beta_n)\text{vec}(\Xbf_i)\Abf^{1/2}_i(\boldsymbol\beta_n)\Fbf_i(\boldsymbol\beta_n)\Kbf_i(\boldsymbol\beta_n)\text{vec}\trans(\Xbf_i)\Jbf(\boldsymbol\beta_n),\\
\tilde{\Dbf}_{n4}(\boldsymbol\beta_n)&=\sum_{i=1}^n\sum_{j=1}^m{\bf e}_j\trans\Abf^{1/2}_i(\boldsymbol\beta_n)\tilde{\Rbf}^{-1}\Abf^{-1/2}_i(\boldsymbol\beta_n)(\Ybf_i-\boldsymbol\mu_i(\boldsymbol\beta_n)){\bf H}(\boldsymbol\beta_n,\Xbf_{ij}),
\end{align*}
with
\vspace{-0.1in}
\begin{align*}
{\bf C}_i(\boldsymbol\beta_n)&=\text{diag}\Big(Y_{i1}-\mu_{i1}(\boldsymbol\beta_n), \dots, Y_{im}-\mu_{im}(\boldsymbol\beta_n)\Big),\\
\Fbf_i(\boldsymbol\beta_n)&=\text{diag}\Big(\mu_{i1}^{(2)}(\boldsymbol\beta_n), \dots, \mu_{im}^{(2)}(\boldsymbol\beta_n)\Big),\\
\Kbf_i(\boldsymbol\beta_n)&=\text{diag}\Big(\tilde{\Rbf}^{-1}\Abf^{-1/2}_i(\boldsymbol\beta_n)(\Ybf_i-\boldsymbol\mu_i(\boldsymbol\beta_n))\Big),
\end{align*}
${\bf e}\trans _j$  length $m$ vector with $j$-th element 1 and 0 everywhere else, and ${\bf H}(\boldsymbol\beta_n,\Xbf_{ij}) = \partial \Jbf\trans (\boldsymbol\beta_n)\text{vec}(\Xbf_{ij}) / \partial \boldsymbol\beta\trans _n$.

Let $\Dbf_{ni}(\boldsymbol\beta_n)$ be defined the same as $\tilde{\Dbf}_{ni}(\boldsymbol\beta_n)$, but with $\tilde{\Rbf}$ replaced by $\widehat{\Rbf}$, for $i=1, \dots, 4$. It is sufficient to prove  
\begin{equation*}
\sup_{||\boldsymbol\beta_n-\boldsymbol\beta_0||\le \triangle \sqrt{p_n/n}}\sup_{\ubf}|\ubf\trans [\Dbf_{ni}(\boldsymbol\beta_n)-\tilde{\Dbf}_{ni}(\boldsymbol\beta_n)]\ubf|=O_p(\sqrt{p_nn})
\end{equation*}
for any $\ubf \in \real{R \sum_{d=1}^D p_d}$ such that $||\ubf||=1$, $i=1, \dots, 4$. 

For $i=1$, we have
\begin{align*}
&|\ubf\trans [\Dbf_{n1}(\boldsymbol\beta_n)-\tilde{\Dbf}_{n1}(\boldsymbol\beta_n)]\ubf|\\
\le& ||\ubf||^2\cdot||\widehat{\Rbf}^{-1}-\tilde{\Rbf}^{-1}||_F\cdot\lambda_{\max}(\Abf_i(\boldsymbol\beta_n))\cdot\lambda_{\max}\Big(\sum_{i=1}^n\Jbf\trans (\boldsymbol\beta_n)\text{vec}(\Xbf_i)\text{vec}\trans (\Xbf_i)\Jbf(\boldsymbol\beta_n)\Big).
\end{align*}
By Condition (A4*) and (A6*), $|\ubf\trans [\Dbf_{n1}(\boldsymbol\beta_n)-\tilde{\Dbf}_{n1}(\boldsymbol\beta_n)]\ubf|=O_p(\sqrt{p_nn})$ on the set $\{\boldsymbol\beta_n:||\boldsymbol\beta_n-\boldsymbol\beta_0||\le \triangle \sqrt{p_n/n}\}$.

For $i=2$, we have
\begin{align*}
&|\ubf\trans [\Dbf_{n2}(\boldsymbol\beta_n)-\tilde{\Dbf}_{n2}(\boldsymbol\beta_n)]\ubf|\\
\le &\frac{1}{2}|\ubf\trans \sum_{i=1}^n\Jbf\trans(\boldsymbol\beta_n)\text{vec}(\Xbf_i)\Abf^{1/2}_i(\boldsymbol\beta_n)(\tilde{\Rbf}^{-1}-\widehat{\Rbf}^{-1})\Abf^{-3/2}_i(\boldsymbol\beta_n){\bf C}_{i1}(\boldsymbol\beta_n)\Fbf_i(\boldsymbol\beta_n)\text{vec}\trans(\Xbf_i)\Jbf(\boldsymbol\beta_n)\ubf|\\
&+\frac{1}{2}|\ubf\trans \sum_{i=1}^n\Jbf\trans(\boldsymbol\beta_n)\text{vec}(\Xbf_i)\Abf^{1/2}_i(\boldsymbol\beta_n)(\tilde{\Rbf}^{-1}-\widehat{\Rbf}^{-1})\Abf^{-3/2}_i(\boldsymbol\beta_n){\bf C}_{i2}(\boldsymbol\beta_0)\Fbf_i(\boldsymbol\beta_n)\text{vec}\trans(\Xbf_i)\Jbf(\boldsymbol\beta_n)\ubf|\\
\triangleq&J_{n1}+J_{n2},
\end{align*}
where 
\begin{align*}
{\bf C}_{i1}(\boldsymbol\beta_n)&=\text{diag}\Big(\mu_{i1}(\boldsymbol\beta_0)-\mu_{i1}(\boldsymbol\beta_n), \dots, \mu_{im}(\boldsymbol\beta_0)-\mu_{im}(\boldsymbol\beta_n)\Big),\\
{\bf C}_{i2}(\boldsymbol\beta_0)&=\text{diag}\Big(Y_{i1}-\mu_{i1}(\boldsymbol\beta_0), \dots, Y_{im}-\mu_{im}(\boldsymbol\beta_0)\Big).
\end{align*} 

For $J_{n1}$, by Cauchy-Schwarz inequality for matrices with Frobenius norm, 
\begin{align*}
J_{n1}\le& C||\widehat{\Rbf}^{-1}-\tilde{\Rbf}^{-1}||_F\cdot\lambda_{\max}\Big(\sum_{i=1}^n\Jbf\trans(\boldsymbol\beta_n)\text{vec}(\Xbf_i)\text{vec}\trans(\Xbf_i)\Jbf(\boldsymbol\beta_n)\Big)\\
&\max_{i,j}|\mu_{ij}^{(1)}(\widetilde{\boldsymbol\beta}_n)|\cdot\mu_{ij}^{(2)}(\boldsymbol\beta_n)|\cdot\frac{\max_j \sigma_{ij}(\boldsymbol\beta_n)}{\min_j \sigma_{ij}(\boldsymbol\beta_n)},
\end{align*}
where $\widetilde{\boldsymbol\beta}_n$ is between $\boldsymbol\beta_n$ and $\boldsymbol\beta_0$. By Conditions (A3*), (A4*) and (A8*), $\sup_{||\boldsymbol\beta_n-\boldsymbol\beta_0||\le \triangle n^{-1/2}}J_{n1}=O_p(\sqrt{p_nn})$.

For $J_{n2}$, we decompose $\Abf^{1/2}_i(\boldsymbol\beta_n)$ into $\Abf^{1/2}_i(\boldsymbol\beta_0)$ and $[\Abf^{1/2}_i(\boldsymbol\beta_n)-\Abf^{1/2}_i(\boldsymbol\beta_0)]$. That is, 
\begin{align*}
2J_{n2}\le& |\ubf\trans \sum_{i=1}^n\Jbf\trans (\boldsymbol\beta_n)\text{vec}(\Xbf_i)\Abf^{1/2}_i(\boldsymbol\beta_0)(\tilde{\Rbf}^{-1}-\widehat{\Rbf}^{-1})\Abf^{-3/2}_i(\boldsymbol\beta_n)\\
&\times {\bf C}_{i2}(\boldsymbol\beta_0)\Fbf_i(\boldsymbol\beta_n)\text{vec}\trans (\Xbf_i)\Jbf(\boldsymbol\beta_n)\ubf|\\
+&|\sum_{i=1}^n\ubf\trans \Jbf\trans (\boldsymbol\beta_n)\text{vec}(\Xbf_i)[\Abf^{1/2}_i(\boldsymbol\beta_n)-\Abf^{1/2}_i(\boldsymbol\beta_0)](\tilde{\Rbf}^{-1}-\widehat{\Rbf}^{-1})\Abf^{-3/2}_i(\boldsymbol\beta_n)\\
&\times {\bf C}_{i2}(\boldsymbol\beta_0)\Fbf_i(\boldsymbol\beta_n)\text{vec}\trans (\Xbf_i)\Jbf(\boldsymbol\beta_n)\ubf|\\
\triangleq&J_{n21}+J_{n22}.
\end{align*}
Similarly to $J_{n1}$, it can be shown $\sup_{||\boldsymbol\beta_n-\boldsymbol\beta_0||\le \triangle \sqrt{p_n/n}}J_{n21}=O_p(\sqrt{np_n})$.

For $J_{n22}$, similar to the decomposition of $\Abf^{1/2}_i(\boldsymbol\beta_n)$, we can further decompose those terms involving $\boldsymbol\beta_n$ into terms that only depend on $\boldsymbol\beta_0$ and four other terms involving $\boldsymbol\beta_n$. On the set $\{\boldsymbol\beta_n:||\boldsymbol\beta_n-\boldsymbol\beta_0||\le \triangle \sqrt{p_n/n}\}$, similar to $J_{n1}$, under Conditions (A1*)-(A9*), all those terms involving $\boldsymbol\beta_n$ can be shown to be $O_p(\sqrt{p_nn})$. To complete the evaluation of $J_{n22}$ and hence $J_{n2}$, it suffices to show
\begin{align} \label{lemm2}
|\sum_{i=1}^n&\ubf\trans \Jbf\trans (\boldsymbol\beta_0)\text{vec}(\Xbf_i)\Abf^{1/2}_i(\boldsymbol\beta_0)(\tilde{\Rbf}^{-1}-\widehat{\Rbf}^{-1})\Abf^{-3/2}_i(\boldsymbol\beta_0) \notag\\
&\times{\bf C}_{i2}(\boldsymbol\beta_0)\Fbf_i(\boldsymbol\beta_0)\text{vec}\trans (\Xbf_i)\Jbf(\boldsymbol\beta_0)\ubf|=O_p(\sqrt{p_nn}).
\end{align}

Denote $L_n(\boldsymbol\beta_0)$ the left side of (\ref{lemm2}). Recall that $\epsilon_{ij}(\boldsymbol\beta_0)=\sigma_{ij}^{-1}(\boldsymbol\beta_0)(Y_{ij}-\mu_{ij}(\boldsymbol\beta_0))$. We have
\begin{align*}
&\mathbb{E}[||L_n(\boldsymbol\beta_0)||^2]=\text{Tr}[\mathbb{E}(L_n(\boldsymbol\beta_0)\trans L_n(\boldsymbol\beta_0))]\\
=&\sum_{i=1}^n\sum_{j=1}^m\sum_{k=1}^m\mathbb{E}[\epsilon_{ij}\epsilon_{ik}]\text{Tr}\Big[\Jbf\trans (\boldsymbol\beta_0)\text{vec}(\Xbf_i)\Abf^{1/2}_i(\boldsymbol\beta_0)(\tilde{\Rbf}^{-1}-\widehat{\Rbf}^{-1})\Abf^{-3/2}_i(\boldsymbol\beta_0){\bf e}_j{\bf e}_j\trans \Fbf_i(\boldsymbol\beta_0)\\
&\cdot\text{vec}\trans (\Xbf_i)\Jbf(\boldsymbol\beta_0)\Jbf\trans (\boldsymbol\beta_0)\text{vec}(\Xbf_i){\bf e}_k{\bf e}_k\trans \Abf^{-3/2}_i(\boldsymbol\beta_0)(\tilde{\Rbf}^{-1}-\widehat{\Rbf}^{-1})\Abf^{1/2}_i(\boldsymbol\beta_0)\text{vec}\trans (\Xbf_i)\Jbf(\boldsymbol\beta_0)\\
\le & C\sum_{i=1}^n\sum_{j=1}^m\sum_{k=1}^m||{\bf e}_j\trans \Fbf_i(\boldsymbol\beta_0)\text{vec}\trans (\Xbf_i)\Jbf(\boldsymbol\beta_0)||\cdot||\Jbf\trans (\boldsymbol\beta_0)\text{vec}(\Xbf_i)\Fbf_i(\boldsymbol\beta_0){\bf e}_k||\\
&\cdot||{\bf e}_k\trans \Abf^{-3/2}_i(\boldsymbol\beta_0)(\tilde{\Rbf}^{-1}-\widehat{\Rbf}^{-1})\Abf^{1/2}_i(\boldsymbol\beta_0)\text{vec}\trans (\Xbf_i)\Jbf(\boldsymbol\beta_0)||\\
&\cdot||\Jbf\trans (\boldsymbol\beta_0)\text{vec}(\Xbf_i)\Abf^{1/2}_i(\boldsymbol\beta_0)(\tilde{\Rbf}^{-1}-\widehat{\Rbf}^{-1})\Abf^{-3/2}_i(\boldsymbol\beta_0){\bf e}_j||.
\end{align*}
By Conditions (A1*), (A2*) and (A4*)-(A7*), $\mathbb{E}[||L_n(\boldsymbol\beta_0)||^2]=O(np_n)$. This implies $J_{n22}=O_p(\sqrt{p_nn})$ and hence 
\begin{equation*}
\sup_{||\boldsymbol\beta_n-\boldsymbol\beta_0||\le \triangle \sqrt{p_n/n}}\sup_{\ubf}|\ubf\trans [\Dbf_{n2}(\boldsymbol\beta_n)-\tilde{\Dbf}_{n2}(\boldsymbol\beta_n)]\ubf|=O_p(\sqrt{p_nn}).
\end{equation*}
Using similar decompositions, we can verify the results for $\Dbf_{n3}$ and $\Dbf_{n4}$, which completes the proof. 
\end{proof}

Based on Lemma \ref{Dbar}, we can further approximate ${\bf  \tilde{D}}_{n}(\boldsymbol\beta_n)$ by ${\bf  \tilde{D}}_{n1}(\boldsymbol\beta_n)$, which is easier to evaluate. Lemma \ref{Dn1} provides this approximation. 
\begin{lemma}\label{Dn1}
Under Conditions (A1*)-(A9*), for any $\triangle>0$ and $\ubf \in \real{R \sum_{d=1}^D p_d}$ such that $||\ubf||=1$, 
\begin{align}
\sup_{||\boldsymbol\beta_n-\boldsymbol\beta_0||=\triangle \sqrt{p_n/n}}\sup_{\ubf}|\ubf\trans [\tilde{\Dbf}_{n}(\boldsymbol\beta_n)-\tilde{\Dbf}_{n1}(\boldsymbol\beta_n)]\ubf|=O_p(n^{1/2}p_n), \label{lemma3.1}\\
\sup_{||\boldsymbol\beta_n-\boldsymbol\beta_0||=\triangle \sqrt{p_n/n}}\sup_{\ubf}|\ubf\trans [\tilde{\Dbf}_{n1}(\boldsymbol\beta_0)-\tilde{\Dbf}_{n1}(\boldsymbol\beta_n)]\ubf|=O_p(n^{1/2}p_n). \label{lemma3.2}
\end{align}
\end{lemma}
\begin{proof}[Proof of Lemma \ref{Dn1}]
To prove (\ref{lemma3.1}), it is sufficient to show, for $i=2, 3, 4$, 
\begin{equation*}
\sup_{||\boldsymbol\beta_n-\boldsymbol\beta_0||=\triangle \sqrt{p_n/n}}\sup_{\ubf}|\ubf\trans \tilde{\Dbf}_{ni}(\boldsymbol\beta_n)\ubf|=O_p(n^{1/2}p_n).
\end{equation*}

For $\tilde{\Dbf}_{n2}(\boldsymbol\beta_n)$, we have the decomposition 
\begin{equation*}
\tilde{\Dbf}_{n2}(\boldsymbol\beta_n)=\tilde{\Dbf}_{n2}(\boldsymbol\beta_0)+\sum_{k=1}^6J_{n6}(\boldsymbol\beta_n),\\
\end{equation*}
where 
\begin{align*}
J_{n1}&=\frac{1}{2}\sum_{i=1}^n[\Jbf\trans (\boldsymbol\beta_n)-\Jbf\trans (\boldsymbol\beta_0)]\text{vec}(\Xbf_i)\Abf^{1/2}_i(\boldsymbol\beta_0)\tilde{\Rbf}^{-1}\Abf^{-3/2}_i(\boldsymbol\beta_0){\bf C}_i(\boldsymbol\beta_0)\Fbf_i(\boldsymbol\beta_0)\text{vec}\trans (\Xbf_i)\Jbf(\boldsymbol\beta_0),\\
J_{n2}&=\frac{1}{2}\sum_{i=1}^n\Jbf\trans (\boldsymbol\beta_n)\text{vec}(\Xbf_i)[\Abf^{1/2}_i(\boldsymbol\beta_n)-\Abf^{1/2}_i(\boldsymbol\beta_0)]\tilde{\Rbf}^{-1}\Abf^{-3/2}_i(\boldsymbol\beta_0){\bf C}_i(\boldsymbol\beta_0)\Fbf_i(\boldsymbol\beta_0)\text{vec}\trans (\Xbf_i)\Jbf(\boldsymbol\beta_0),\\
J_{n3}&=\frac{1}{2}\sum_{i=1}^n\Jbf\trans (\boldsymbol\beta_n)\text{vec}(\Xbf_i)\Abf^{1/2}_i(\boldsymbol\beta_n)\tilde{\Rbf}^{-1}[\Abf^{-3/2}_i(\boldsymbol\beta_n)-\Abf^{-3/2}_i(\boldsymbol\beta_0)]{\bf C}_i(\boldsymbol\beta_0)\Fbf_i(\boldsymbol\beta_0)\text{vec}\trans (\Xbf_i)\Jbf(\boldsymbol\beta_0),\\
J_{n4}&=\frac{1}{2}\sum_{i=1}^n\Jbf\trans (\boldsymbol\beta_n)\text{vec}(\Xbf_i)\Abf^{1/2}_i(\boldsymbol\beta_n)\tilde{\Rbf}^{-1}\Abf^{-3/2}_i(\boldsymbol\beta_0)[{\bf C}_i(\boldsymbol\beta_n)-{\bf C}_i(\boldsymbol\beta_0)]\Fbf_i(\boldsymbol\beta_0)\text{vec}\trans (\Xbf_i)\Jbf(\boldsymbol\beta_0),\\
J_{n5}&=\frac{1}{2}\sum_{i=1}^n\Jbf\trans (\boldsymbol\beta_n)\text{vec}(\Xbf_i)\Abf^{1/2}_i(\boldsymbol\beta_n)\tilde{\Rbf}^{-1}\Abf^{-3/2}_i(\boldsymbol\beta_n){\bf C}_i(\boldsymbol\beta_n)[\Fbf_i(\boldsymbol\beta_n)-\Fbf_i(\boldsymbol\beta_0)]\text{vec}\trans (\Xbf_i)\Jbf(\boldsymbol\beta_0),\\
J_{n6}&=\frac{1}{2}\sum_{i=1}^n\Jbf\trans (\boldsymbol\beta_n)\text{vec}(\Xbf_i)\Abf^{1/2}_i(\boldsymbol\beta_n)\tilde{\Rbf}^{-1}\Abf^{-3/2}_i(\boldsymbol\beta_n){\bf C}_i(\boldsymbol\beta_n)\Fbf_i(\boldsymbol\beta_n)\text{vec}\trans (\Xbf_i)[\Jbf(\boldsymbol\beta_n)-\Jbf(\boldsymbol\beta_0)].
\end{align*}
Using the same techniques as in Lemma \ref{Dbar}, it can be shown that $\tilde{\Dbf}_{n2}(\boldsymbol\beta_0)$ and $J_{ni}$ are $O_p(n^{1/2}p_n)$, $i=1, \dots, 6$. We can prove $\tilde{\Dbf}_{n3}(\boldsymbol\beta_0)$ and $\tilde{\Dbf}_{n4}(\boldsymbol\beta_0)$ in the same way, which completes the proof of (\ref{lemma3.1}).

To prove (\ref{lemma3.2}), note that
\begin{align*}
&|\ubf\trans [\tilde{\Dbf}_{n1}(\boldsymbol\beta_0)-\tilde{\Dbf}_{n1}(\boldsymbol\beta_n)]\ubf|\\
\le& |\ubf\trans [\Jbf\trans (\boldsymbol\beta_0)-\Jbf\trans (\boldsymbol\beta_n)]\text{vec}(\Xbf_i)\Abf^{1/2}_i(\boldsymbol\beta_n)\tilde{\Rbf}^{-1}\Abf^{1/2}_i(\boldsymbol\beta_n)\text{vec}\trans (\Xbf_i)\Jbf(\boldsymbol\beta_n)\ubf|\\
&+ |\ubf\trans \Jbf\trans (\boldsymbol\beta_0)\text{vec}(\Xbf_i)[\Abf^{1/2}_i(\boldsymbol\beta_0)-\Abf^{1/2}_i(\boldsymbol\beta_n)]\tilde{\Rbf}^{-1}\Abf^{1/2}_i(\boldsymbol\beta_n)\text{vec}\trans (\Xbf_i)\Jbf(\boldsymbol\beta_n)\ubf|\\
&+ |\ubf\trans \Jbf\trans (\boldsymbol\beta_0)\text{vec}(\Xbf_i)\Abf^{1/2}_i(\boldsymbol\beta_0)\tilde{\Rbf}^{-1}[\Abf^{1/2}_i(\boldsymbol\beta_0)-\Abf^{1/2}_i(\boldsymbol\beta_n)]\text{vec}\trans (\Xbf_i)\Jbf(\boldsymbol\beta_n)\ubf|\\
&+ |\ubf\trans \Jbf\trans (\boldsymbol\beta_0)\text{vec}(\Xbf_i)\Abf^{1/2}_i(\boldsymbol\beta_0)\tilde{\Rbf}^{-1}\Abf^{1/2}_i(\boldsymbol\beta_0)\text{vec}\trans (\Xbf_i)[\Jbf(\boldsymbol\beta_0)-\Jbf(\boldsymbol\beta_n)]\ubf|.
\end{align*}
The rest of the proof is similar to Lemma \ref{Dbar} and thus is omitted here. \end{proof}

\noindent
\textbf{Proof of Theorem 3}
\noindent
\begin{proof}
\citet{Wang2011} gave a sufficient condition for the existence and consistency of a sequence of root $\widehat{\boldsymbol\beta}_n$of $\sbf_n(\boldsymbol\beta_n)=0$, namely, 
\begin{equation} \label{suff}
P(\sup_{||\boldsymbol\beta_n-\boldsymbol\beta_0||=\triangle \sqrt{p_n/n}}(\boldsymbol\beta_n-\boldsymbol\beta_0)\trans \sbf_n(\boldsymbol\beta_n)<0)\ge 1-\epsilon
\end{equation}
with $\forall \epsilon>0$ and a constant $\triangle>0$. 
To verify ($\ref{suff}$), the main idea is to approximate $\sbf_n(\boldsymbol\beta_n)$ by $\tilde{\sbf}_n(\boldsymbol\beta_n)$, whose moments are easier to evaluate. 

By direct calculation, 
\begin{eqnarray*}
(\boldsymbol\beta_n-\boldsymbol\beta_0)\trans \sbf_n(\boldsymbol\beta_n) 
&=&(\boldsymbol\beta_n-\boldsymbol\beta_0)\trans \sbf_n(\boldsymbol\beta_0)-(\boldsymbol\beta_n-\boldsymbol\beta_0)\trans \Dbf_n(\boldsymbol\beta^*_n)(\boldsymbol\beta_n-\boldsymbol\beta_0)\\
& \triangleq & I_{n1}+I_{n2},
\end{eqnarray*}
where $\boldsymbol\beta^*_n=t\boldsymbol\beta_n+(1-t)\boldsymbol\beta^0$ for some $0<t<1$. Further decompose $I_{n1}$ into 
\begin{align*}
I_{n1}&=(\boldsymbol\beta_n-\boldsymbol\beta_0)\trans \tilde{\sbf}_n(\boldsymbol\beta_0)+(\boldsymbol\beta_n-\boldsymbol\beta_0)\trans [\sbf_n(\boldsymbol\beta_0)-\tilde{\sbf}_n(\boldsymbol\beta_0)]\\
&\triangleq I_{n11}+I_{n12}.
\end{align*}
Note that $I_{n11}\le \triangle \sqrt{p_n/n}\cdot||\tilde{\sbf}_n(\boldsymbol\beta_0)||$, where
\begin{align*}
&\mathbb{E}[||\tilde{\sbf}_n(\boldsymbol\beta_0)||^2]\\
=&\mathbb{E}\Big\{\sum_{i=1}^n\boldsymbol\epsilon_i\trans \tilde{\Rbf}^{-1}\Abf^{1/2}_i(\boldsymbol\beta_0)\text{vec}\trans (\Xbf_i)\Jbf(\boldsymbol\beta_n)\Jbf\trans (\boldsymbol\beta_n)\text{vec}(\Xbf_i)\Abf^{1/2}_i(\boldsymbol\beta_0)\tilde{\Rbf}^{-1}\boldsymbol\epsilon_i\Big\}\\
\le &C\cdot\text{Tr}\Big(\Jbf\trans (\boldsymbol\beta_n)\text{vec}(\Xbf_i)\text{vec}\trans (\Xbf_i)\Jbf(\boldsymbol\beta_n)\Big)\\
=&C\sum_{j=1}^m\cdot\text{Tr}\Big(\Jbf\trans (\boldsymbol\beta_n)\text{vec}(\Xbf_{ij})\text{vec}\trans (\Xbf_{ij})\Jbf(\boldsymbol\beta_n)\Big)=O(np_n)
\end{align*}
for some constant $C>0$. This implies that $I_{n11}=\triangle \sqrt{p_n/n}O_p(\sqrt{np_n})=\triangle O_p(p_n)$. For $I_{n12}$, by Lemma \ref{Sbar},
\begin{equation*}
I_{n12}\le||\boldsymbol\beta_n-\boldsymbol\beta_0||\cdot||\sbf_n(\boldsymbol\beta_0)-\tilde{\sbf}_n(\boldsymbol\beta_0)||=o_p(p_n).
\end{equation*}
Therefore, $I_{n1}$ is dominated in probability by $I_{n11}$.

For $I_{n2}$ ,we decompose it into 
\begin{align*}
I_{n2}=&-(\boldsymbol\beta_n-\boldsymbol\beta_0)\trans \tilde{\Dbf}_n(\boldsymbol\beta^*_n)(\boldsymbol\beta_n-\boldsymbol\beta_0)\\
&-(\boldsymbol\beta_n-\boldsymbol\beta_0)\trans [\Dbf_n(\boldsymbol\beta^*_n)-\tilde{\Dbf}_n(\boldsymbol\beta^*_n)](\boldsymbol\beta_n-\boldsymbol\beta_0)\\
\triangleq&I_{n21}+I_{n22}.
\end{align*} 
By Lemma \ref{Dbar}, it can be easily checked that $I_{n22}=o_p(p_n)$. Next for $I_{n21}$,
\begin{align*}
I_{n21}=&-(\boldsymbol\beta_n-\boldsymbol\beta_0)\trans \tilde{\Dbf}_{n1}(\boldsymbol\beta_0)(\boldsymbol\beta_n-\boldsymbol\beta_0)\\
&-(\boldsymbol\beta_n-\boldsymbol\beta_0)\trans [\tilde{\Dbf}_{n1}(\boldsymbol\beta^*_n)-\tilde{\Dbf}_{n1}(\boldsymbol\beta_0)](\boldsymbol\beta_n-\boldsymbol\beta_0)\\
&-(\boldsymbol\beta_n-\boldsymbol\beta_0)\trans [\tilde{\Dbf}_{n}(\boldsymbol\beta^*_n)-\tilde{\Dbf}_{n1}(\boldsymbol\beta^*_n)](\boldsymbol\beta_n-\boldsymbol\beta_0)\\
\triangleq&I^1_{n21}+I^2_{n21}+I^3_{n21}.
\end{align*}

We next show that $I_{n21}$ is dominated in probability by $I^1_{n21}$. Note that by Condition (A3*), (A4*) and (A8*), 
\begin{align*}
I^1_{n21}=&-(\boldsymbol\beta_n-\boldsymbol\beta_0)\trans \Big[\sum_{i=1}^n\Jbf\trans (\boldsymbol\beta_n)\text{vec}(\Xbf_i)\Abf^{1/2}_i(\boldsymbol\beta_n)\tilde{\Rbf}^{-1}\Abf^{1/2}_i(\boldsymbol\beta_n)
\text{vec}\trans (\Xbf_i)\Jbf(\boldsymbol\beta_n)\Big](\boldsymbol\beta_n-\boldsymbol\beta_0)\\
\le&-n^{-1}\triangle^2\min_i\lambda_{\min}(\Abf_i(\boldsymbol\beta_n))\lambda_{\min}\Big(\sum_{i=1}^n\Jbf\trans (\boldsymbol\beta_n)\text{vec}(\Xbf_i)\text{vec}\trans (\Xbf_i)\Jbf(\boldsymbol\beta_n)\Big)\lambda_{\min}(\tilde{\Rbf}^{-1})\\
\le &-C\triangle^2p_n,
\end{align*}
for some constant $C>0$. By Lemma \ref{Dn1}, it can be checked directly that both $I^2_{n21}$ and $I^3_{n21}$ are $o_p(p_n)$. 

Therefore, the sign of $(\boldsymbol\beta_n-\boldsymbol\beta_0)\trans \sbf_n(\boldsymbol\beta_n)$ is determined by in probability by $I_{n11}+I^1_{n21}$ and is negative for sufficiently large $\triangle$, which completes the proof.\end{proof}

\noindent
\textbf{Proof of Theorem 4}
\noindent
\begin{proof}
We first show that the normalized $\tilde{\sbf}_n(\boldsymbol\beta_0)$ has an asymptotic normal distribution. That is, for any $\bbf \in \real{R \sum_{d=1}^D p_d}$ such that $||\bbf||=1$, 
\begin{equation} \label{CLT}
\bbf\trans \tilde{\Mbf}^{-1/2}_n(\boldsymbol\beta_0)\tilde{\sbf}_n(\boldsymbol\beta_0)\to N(0, 1),
\end{equation}
where $\tilde{\Mbf}_n(\boldsymbol\beta_0)=\text{Var}(\tilde{\sbf}_n(\boldsymbol\beta_0))$. 

Denote $\bbf\trans \tilde{\Mbf}^{-1/2}_n(\boldsymbol\beta_0)\tilde{\sbf}_n(\boldsymbol\beta_0)=\sum_{i=1}^nZ_{ni}$, where 
\begin{equation*}
Z_{ni}=\bbf\trans \tilde{\Mbf}^{-1/2}_n(\boldsymbol\beta_0)\Jbf\trans (\boldsymbol\beta_0)\text{vec}(\Xbf_i)\Abf^{1/2}_i(\boldsymbol\beta_0)\tilde{\Rbf}^{-1}\boldsymbol\epsilon_i(\boldsymbol\beta_0),
\end{equation*}
where $\boldsymbol\epsilon_i(\boldsymbol\beta_0)=\Abf^{-1/2}_i(\boldsymbol\beta_0)({\bf Y}_i-\boldsymbol\mu_i(\boldsymbol\beta_0))$. Note that $\mathbb{E}(Z_{ni})=0$, $\text{Var}(\sum_{i=1}^nZ_{ni})=1$. To prove (\ref{CLT}), it suffices to check the Lyapunov condition. That is, for some $\delta>0$,
\begin{equation*}
\sum_{i=1}^n\mathbb{E}\Big(|Z_{ni}|^{2+\delta}\Big) \to 0,
\end{equation*}
as $n \to \infty$. By Cauchy-Schwarz inequality, 
\begin{align*}
Z_{ni}^2\le \lambda_{\max}(\tilde{\Rbf}^{-2})\lambda_{\max}(\Abf_i(\boldsymbol\beta_0))||\boldsymbol\epsilon_i(\boldsymbol\beta_0)||^2\gamma_{ni},
\end{align*}
where $\gamma_{ni}\triangleq \bbf\trans \tilde{\Mbf}^{-1/2}_n(\boldsymbol\beta_0)\Jbf\trans (\boldsymbol\beta_0)\text{vec}(\Xbf_i)\text{vec}\trans (\Xbf_i)\Jbf(\boldsymbol\beta_0)\tilde{\Mbf}^{-1/2}_n(\boldsymbol\beta_0)\bbf$. To evaluate $\max_{1 \le i \le n}\gamma_{ni}$, we need to evaluate $\lambda^{-1}_{\min}(\tilde{\Mbf}_n(\boldsymbol\beta_0))$. Note that
\begin{align*}
\bbf\trans \tilde{\Mbf}_n(\boldsymbol\beta_0)\bbf \ge &C\bbf\trans \Big(\sum_{i=1}^n\Jbf\trans (\boldsymbol\beta_0)\text{vec}(\Xbf_i)\text{vec}\trans (\Xbf_i)\Jbf(\boldsymbol\beta_0)\Big)\bbf\\
\ge &C\lambda_{\min}\Big(\sum_{i=1}^n\Jbf\trans (\boldsymbol\beta_0)\text{vec}(\Xbf_i)\text{vec}\trans (\Xbf_i)\Jbf(\boldsymbol\beta_0)\Big).
\end{align*}
By Condition (A1*) and (A3*), $\lambda^{-1}_{\min}(\tilde{\Mbf}_n(\boldsymbol\beta_0))=O(n^{-1})$ and hence $\max_{1 \le i \le n}\gamma_{ni}=o(1)$. 

It follows that for any $\delta>0$,
\begin{align*}
\sum_{i=1}^n\mathbb{E}\Big(|Z_{ni}|^{2+\delta}\Big)&\le \sum_{i=1}^n\mathbb{E}\Big(C^{1+\delta/2}\gamma_{ni}^{1+\delta/2}||\boldsymbol\epsilon_i(\boldsymbol\beta_0)||^{2+\delta}\Big)\\ 
&\le C(\max_{1 \le i \le n}\gamma_{ni})^{\delta/2}\sum_{i=1}^n\bbf\trans \tilde{\Mbf}^{-1/2}_n(\boldsymbol\beta_0)\Jbf\trans (\boldsymbol\beta_0)\text{vec}(\Xbf_i)\text{vec}\trans (\Xbf_i)\Jbf(\boldsymbol\beta_0)\tilde{\Mbf}^{-1/2}_n(\boldsymbol\beta_0)\bbf\\
&\le C(\max_{1 \le i \le n}\gamma_{ni})^{\delta/2}\lambda_{\max}\big(\sum_{i=1}^n\Jbf\trans (\boldsymbol\beta_0)\text{vec}(\Xbf_i)\text{vec}\trans (\Xbf_i)\Jbf(\boldsymbol\beta_0)\big)\lambda^{-1}_{\min}\big(\tilde{\Mbf}_n(\boldsymbol\beta_0)\big)\\
&=o(1)O(n)O(n^{-1})=o(1),
\end{align*}
which completes the proof of (\ref{CLT}). 

To prove Theorem 2, note that by the fact $\sbf_n(\widehat{\boldsymbol\beta}_n)=0$, we have $\sbf_n(\boldsymbol\beta_0)=\Dbf_n(\boldsymbol\beta^*_n)(\widehat{\boldsymbol\beta}_n-\boldsymbol\beta_0)$ for some $\boldsymbol\beta^*_n$ between $\widehat{\boldsymbol\beta}_n$ and $\boldsymbol\beta_0$. Hence,
\begin{align*}
&\bbf\trans \tilde{\Mbf}^{-1/2}_n(\boldsymbol\beta_0)\tilde{\sbf}_n(\boldsymbol\beta_0)\\
=&\bbf\trans \tilde{\Mbf}^{-1/2}_n(\boldsymbol\beta_0)\tilde{\Dbf}_{n1}(\boldsymbol\beta_0)(\widehat{\boldsymbol\beta}_n-\boldsymbol\beta_0)\\
&+\bbf\trans \tilde{\Mbf}^{-1/2}_n(\boldsymbol\beta_0)[\Dbf_n(\boldsymbol\beta^*_n)-\tilde{\Dbf}_{n1}(\boldsymbol\beta_0)](\widehat{\boldsymbol\beta}_n-\boldsymbol\beta_0)\\
&+\bbf\trans \tilde{\Mbf}^{-1/2}_n(\boldsymbol\beta_0)[\tilde{\sbf}_n(\boldsymbol\beta_0)-\sbf_n(\boldsymbol\beta_0)]\\
=&J_{n1}+J_{n2}(\boldsymbol\beta^*_n)+J_{n3}(\boldsymbol\beta_0).
\end{align*}
By (\ref{CLT}), it is sufficient to prove that both $\sup_{||\boldsymbol\beta_n-\boldsymbol\beta_0||\le \triangle \sqrt{p_n/n}}|J_{n2}(\boldsymbol\beta_n)|$ and $|J_{n3}(\boldsymbol\beta_0)|$ are $o_p(1)$. 

For $J_{n3}$, recall that $||\tilde{\sbf}_n(\boldsymbol\beta_0)-\sbf_n(\boldsymbol\beta_0)||=O_p(1)$ from Lemma \ref{Sbar}. Using the previous result that $\lambda^{-1}_{\min}(\tilde{\Mbf}_n(\boldsymbol\beta_0))=O(n^{-1})$, it can be easily checked that $J_{n3}^2=o_p(1)$ and hence $|J_{n3}|=o_p(1)$. 

For $J_{n2}$, we have
\begin{align*}
&\sup_{||\boldsymbol\beta_n-\boldsymbol\beta_0||\le \triangle \sqrt{p_n/n}}|J_{n2}(\boldsymbol\beta_n)|\\
\le &\sup_{||\boldsymbol\beta_n-\boldsymbol\beta_0||\le \triangle \sqrt{p_n/n}}\bbf\trans \tilde{\Mbf}^{-1/2}_n(\boldsymbol\beta_0)[\Dbf_n(\boldsymbol\beta_n)-\tilde{\Dbf}_{n}(\boldsymbol\beta_n)](\widehat{\boldsymbol\beta}_n-\boldsymbol\beta_0)\\
&+\sup_{||\boldsymbol\beta_n-\boldsymbol\beta_0||\le \triangle \sqrt{p_n/n}}\bbf\trans \tilde{\Mbf}^{-1/2}_n(\boldsymbol\beta_0)[\tilde{\Dbf}_n(\boldsymbol\beta_n)-\tilde{\Dbf}_{n1}(\boldsymbol\beta_n)](\widehat{\boldsymbol\beta}_n-\boldsymbol\beta_0)\\
&+\sup_{||\boldsymbol\beta_n-\boldsymbol\beta_0||\le \triangle \sqrt{p_n/n}}\bbf\trans \tilde{\Mbf}^{-1/2}_n(\boldsymbol\beta_0)[\tilde{\Dbf}_{n1}(\boldsymbol\beta_n)-\tilde{\Dbf}_{n1}(\boldsymbol\beta_0)](\widehat{\boldsymbol\beta}_n-\boldsymbol\beta_0)\\
\triangleq&I_{n1}+I_{n2}+I_{n3}.
\end{align*}
By Theorem 1, Lemma \ref{Dn1} and the fact $\lambda^{-1}_{\min}(\tilde{\Mbf}_n(\boldsymbol\beta_0))=O(n^{-1})$, it can be seen that 
\begin{align*}
I_{n1}\le &C\cdot\lambda_{\max}(\Dbf_n(\boldsymbol\beta_n)-\tilde{\Dbf}_{n}(\boldsymbol\beta_n))\lambda^{-1/2}_{\min}(\tilde{\Mbf}_n(\boldsymbol\beta_0))\sqrt{p_n/n}\\
=&O_p(\sqrt{n}p_n)O(n^{-1/2})O_p(\sqrt{p_n/n})=O_p(p_n^{3/2}n^{-1/2}).
\end{align*}
Using the stronger assumption that $p=o(n^{-1/3})$, $I_{n1}=o_p(1)$. Similarly, by Lemma \ref{Dn1}, we have $I_{n2}=o_p(1)$ and $I_{n3}=o_p(1)$. Therefore $J_{n1}$ has the same  asymptotic distribution as in ($\ref{CLT}$), which completes the proof.  
\end{proof}

\end{document}